\documentclass{article}

\title{Optimal Oracles for Point-to-Set Principles}
\author{D. M. Stull\\Department of Computer Science, Iowa State University\\
	Ames, IA 50011, USA\\
	\texttt{dstull@iastate.edu}}

\date{}

\usepackage{amsmath}
\usepackage{amssymb}
\usepackage{amsthm} 
\usepackage{caption}

\newtheorem{thm}{Theorem}
\newtheorem{obs}[thm]{Observation}
\newtheorem{prop}[thm]{Proposition}
\newtheorem{cor}[thm]{Corollary}
\newtheorem{lem}[thm]{Lemma}

\newtheorem{defn}[thm]{Definition}
\theoremstyle{remark}
\newtheorem*{remark}{Remark}
\newtheorem*{ex}{Example}

\DeclareMathOperator{\Dim}{Dim}

\newcommand{\R}{\mathbb{R}}

\newcommand{\N}{\mathbb{N}}
\newcommand{\Q}{\mathbb{Q}}
\newcommand{\D}{\mathbb{D}}
\newcommand{\ve}{\varepsilon}

\newcommand{\Hs}{\mathcal{H}^s}
\newcommand{\Ps}{\mathcal{P}^s}

\begin{document}
	\maketitle

\begin{abstract}
The point-to-set principle \cite{LutLut17} characterizes the Hausdorff dimension of a subset $E\subseteq\R^n$ by the \textit{effective} (or algorithmic) dimension of its individual points. This characterization has been used to prove several results in classical, i.e., without any computability requirements, analysis. Recent work has shown that algorithmic techniques can be fruitfully applied to Marstrand's projection theorem, a fundamental result in fractal geometry.

In this paper, we introduce an extension of point-to-set principle - the notion of \textit{optimal oracles} for subsets $E\subseteq\R^n$. One of the primary motivations of this definition is that, if $E$ has optimal oracles, then the conclusion of Marstrand's projection theorem holds for $E$. We show that every analytic set has optimal oracles. We also prove that if the Hausdorff and packing dimensions of $E$ agree, then $E$ has optimal oracles. Moreover, we show that the existence of sufficiently nice outer measures on $E$ implies the existence of optimal Hausdorff oracles. In particular, the existence of exact gauge functions for a set $E$ is sufficient for the existence of optimal Hausdorff oracles, and is therefore sufficient for Marstrand's theorem. Thus, the existence of optimal oracles extends the currently known sufficient conditions for Marstrand's theorem to hold.

Under certain assumptions, every set has optimal oracles. However, assuming the axiom of choice and the continuum hypothesis, we construct sets which do not have optimal oracles. This construction naturally leads to a generalization of Davies theorem on projections.
\end{abstract}

\section{Introduction}
Effective, i.e., algorithmic, dimensions were introduced \cite{Lutz03a, AthHitLutMay07} to study the randomness of points in Euclidean space. The effective dimension, $\dim(x)$ and effective strong dimension, $\Dim(x)$, are real values which measure the asymptotic density of information of an \textit{individual point} $x$.  The connection between effective dimensions and the classical Hausdorff and packing dimension is given by the point-to-set principle of J. Lutz and N. Lutz \cite{LutLut17}: For any $E \subseteq \R^n$,
\begin{align}
\dim_H(E) &= \min\limits_{A \subseteq \N} \sup_{x \in E} \dim^A(x), \text{ and} \label{eq:h}\\ 
\dim_P(E) &= \min\limits_{A \subseteq \N} \sup_{x \in E} \Dim^A(x)\label{eq:p}\,.
\end{align}
Call an oracle $A$ satisfying (\ref{eq:h}) a \textit{Hausdorff oracle} for $E$. Similarly, we call an oracle $A$ satisfying (\ref{eq:p}) a \textit{packing oracle} for $E$. Thus, the point-to-set principle shows that the classical notion of Hausdorff or packing dimension is completely characterized by the effective dimension of its individual points, relative to a Hausdorff or packing oracle, respectively.

Recent work as shown that algorithmic dimensions are not only useful in effective settings, but, via the point-to-set principle, can be used to solve problems in geometric measure theory \cite{LutLut20, Lutz17, LutStu17, Stull18}. It is important to note that the point-to-set principle allows one to use algorithmic techniques to prove theorems whose statements have seemingly nothing to do with computability theory. In this paper, we focus on the connection between algorithmic dimension and Marstrand's projection theorem.

Marstrand, in his landmark paper \cite{Marstrand54}, was the first to study how the dimension of a set is changed when projected onto a line. He showed that, for any \textit{analytic} set $E \in \R^2$, for almost every angle $\theta \in [0, \pi)$,
\begin{equation}\label{eq:MarstrandTheorem}
\dim_H(p_\theta \, E) = \min\{\dim_H(E), 1\},
\end{equation}
where $p_\theta(x,y) = x \cos \theta + y\sin \theta$\footnote{This result was later generalized to $\R^n$, for arbitrary $n$, as well as extended to hyperspaces of dimension $m$, for any $1 \leq m \leq n$ (see e.g. \cite{Mattila75, Mattila99, Mattila04}).}. The study of projections has since become a central theme in fractal geometry (see \cite{FalFraJin15} or \cite{Mattila19} for a more detailed survey of this development). 

Marstrand's theorem begs the question of whether the analytic requirement on $E$ can be dropped. It is known that, without further conditions, it cannot. Davies~\cite{Davies79} showed that, assuming the axiom of choice and the continuum hypothesis, there are non-analytic sets for which Marstrands conclusion fails. However, the problem of classifying the sets for which Marstrands theorem does hold is still open. Recently, Lutz and Stull \cite{LutStu18} used the point-to-set principle to prove that the projection theorem holds for sets for which the Hausdorff and packing dimensions agree\footnote{Orponen \cite{Orponen20a} has recently given another proof of Lutz and Stull's result using more classical tools.}. This expanded the reach of Marstrand's theorem, as this assumption is incomparable with analyticity.

In this paper, we give the broadest known sufficient condition (which makes essential use of computability theory) for Marstrand's theorem. In particular, we introduce the notion of \textit{optimal Hausdorff oracles} for a set $E\subseteq \R^n$. We prove that Marstrand's theorem holds for every set $E$ which has optimal Hausdorff oracles. 

An optimal Hausdorff oracle for a set $E$ is a Hausdorff oracle which minimizes the algorithmic complexity of "most``\footnote{By most, we mean a subset of $E$ of the same Hausdorff dimension as $E$} points in $E$. It is not immediately clear that any set $E$ has optimal oracles. Nevertheless, we show that two natural classes of sets $E \subseteq \R^n$ do have optimal oracles. 

We show that every analytic, and therefore Borel, set has optimal oracles.  We also prove that every set whose Hausdorff and packing dimensions agree has optimal Hausdorff oracles. Thus, we show that the existence of optimal oracles encapsulates the known conditions sufficient for Marstrand's theorem to hold. Moreover, we show that the existence of sufficiently nice outer measures on $E$ implies the existence of optimal Hausdorff oracles. In particular, the existence of exact gauge functions (Section \ref{ssection:outerMeasure}) for a set $E$ is sufficient for the existence of optimal Hausdorff oracles for $E$, and is therefore sufficient for Marstrand's theorem. Thus, the existence of optimal Hausdorff oracles is weaker than the previously known conditions for Marstrand's theorem to hold. 

We also show that the notion of optimal oracles gives insight to sets for which Marstrand's theorem does \textit{not} hold. Assuming the axiom of choice and the continuum hypothesis, we construct sets which do not have optimal oracles. This construction, with minor adjustments, proves a generalization of Davies theorem proving the existence of sets for which (\ref{eq:MarstrandTheorem}) does not hold. In addition, the inherently algorithmic aspect of the construction might be useful for proving set-theoretic properties of exceptional sets for Marstrand's theorem.

Finally, we define optimal \textit{packing} oracles for a set. We show that every analytic set $E$ has optimal packing oracles. We also show that every $E$ whose Hausdorff and packing dimensions agree have optimal packing oracles. Assuming the axiom of choice and the continuum hypothesis, we show that there are sets with optimal packing oracles without optimal Hausdorff oracles (and vice-versa).

The structure of the paper is as follows. In Section \ref{ssection:outerMeasure} we review the concepts of measure theory needed, and the (classical) definition of Hausdorff dimension. In Section \ref{ssec:ait} we review algorithmic information theory, including the formal definitions of effective dimensions. We then introduce and study the notion of optimal oracles in Section \ref{sec:optimaloracles}. In particular, we give a general condition for the existence of optimal oracles in Section \ref{ssec:optoraclesouter}. We use this condition to prove that analytic sets have optimal oracles in Section \ref{ssec:setswithopt}. We conclude in Section \ref{ssec:counterexamOptimal} with an example, assuming the axiom of choice and the continuum hypothesis, of a set without optimal oracles. The connection between Marstrands projection theorem and optimal oracles is explored in Section \ref{sec:Marstrand}. In this section, we prove that Marstrands theorem holds for every set with optimal oracles. In Section \ref{ssec:daviesCounter}, we use the construction of a set without optimal oracles to give a new, algorithmic, proof of Davies theorem. Finally, in Sectino \ref{sec:optimalPacking}, we define and investigate the notion of optimal packing oracles.

\section{Preliminaries}

\subsection{Outer Measures and Classical Dimension}\label{ssection:outerMeasure}
A set function $\mu: \mathcal{P}(\R^n) \to [0,\infty]$ is called an \textit{outer measure} on $\R^n$ if 
\begin{enumerate}
\item $\mu(\emptyset) = 0$, 
\item if $A\subseteq B$ then $\mu(A) \leq \mu(B)$, and
\item for any sequence $A_1,A_2,\ldots$ of subsets,
\begin{center}
$\mu(\bigcup_i A_i) \leq \sum_i \mu(A_i)$.
\end{center}
\end{enumerate}
If $\mu$ is an outer measure, we say that a subset $A$ is \textit{$\mu$-measurable} if
\begin{center}
$\mu(A \cap B) + \mu(B - A) = \mu(B)$,
\end{center}
for every subset $B\subseteq \R^n$.

An outer measure $\mu$ is called a \textit{metric outer measure} if every Borel subset is $\mu$-measurable and 
\begin{center}
$\mu(A \cup B) = \mu(A) + \mu(B)$,
\end{center}
for every pair of subsets $A,B$ which have positive Hausdorff distance. That is, 
\begin{center}
$\inf\{ \|x - y\| \, | \, x\in A, y\in B\} > 0$.
\end{center}

An important example of a metric outer measure is the $s$-dimensional Hausdorff measure. For every $E\subseteq [0,1)^n$, define the $s$-dimensional Hausdorff content at precision $r$ by
\begin{center}
$h^s_r(E) = \inf\left\{ \sum_i d(Q_i)^s \, |\, \bigcup_i Q_i \text{ covers } E \text{ and } d(Q_i) \leq 2^{-r}\right\}$,
\end{center}
where $d(Q)$ is the diameter of ball $Q$. We define the $s$-dimensional Hausdorff measure of $E$ by
\begin{center}
$\mathcal{H}^s(E) = \lim\limits_{r\to \infty} h^s_r(E)$.
\end{center}

\begin{remark}
It is well-known that $\Hs$ is a metric outer measure for every $s$.
\end{remark}

The \textit{Hausdorff dimension} of a set $E$ is then defined by 
\begin{center}
$\dim_H(E) = \inf\limits_{s}\{ \Hs(E) = \infty \} = \sup\limits_s \{\Hs(E) = 0\}$.
\end{center}

Another important metric outer measure, which gives rise to the packing dimension of a set, is the $s$-dimensional packing measure. For every $E\subseteq [0,1)^n$, define the $s$-dimensional packing pre-measure by
\begin{center}
$p^s(E) = \limsup\limits_{\delta \to 0} \left\{ \sum\limits_{i\in\N} d(B_i)^s \, |\, \{B_i\} \text{ is a set of disjoint balls and } B_i \in C(E,\delta)\right\}$,
\end{center}
where $C(E,\delta)$ is the set of all closed balls with diameter at most $\delta$ with centers in $E$. We define the $s$-dimensional packing measure of $E$ by
\begin{center}
$\mathcal{P}^s(E) = inf\left\{ \sum\limits_j p^s(E_j) \, | \, E \subseteq \bigcup E_j\right\}$,
\end{center}
where the infimum is taken over all countable covers of $E$. For every $s$, the $s$-dimensional packing measure is a metric outer measure.

The \textit{packing dimension} of a set $E$ is then defined by 
\begin{center}
$\dim_P(E) = \inf\limits_{s}\{ \Ps(E) = 0 \} = \sup\limits_s \{\Ps(E) = \infty\}$.
\end{center}

In order to prove that every analytic set has optimal oracles, we will make use of the following facts of geometric measure theory (see, e.g., \cite{Falconer14}, \cite{BisPer17}).
\begin{thm}\label{thm:compactSset}
The following are true.
\begin{enumerate}
\item Suppose $E \subseteq \R^n$ is compact and satisfies $\mathcal{H}^s(E) > 0$. Then there is a compact subset $F\subseteq E$ such that $0< \Hs(F) <\infty$.
\item Every analytic set $E\subseteq \R^n$ has a $\Sigma^0_2$ subset $F \subseteq E$ such that $\dim_H(F) = \dim_H(E)$.
\item Suppose $E \subseteq \R^n$ is compact and satisfies $\Ps(E) > 0$. Then there is a compact subset $F\subseteq E$ such that $0< \Ps(F) <\infty$.
\item Every analytic set $E\subseteq \R^n$ has a $\Sigma^0_2$ subset $F \subseteq E$ such that $\dim_P(F) = \dim_P(E)$.
\end{enumerate}
\end{thm}

It is possible to generalize the definition of Hausdorff measure using gauge functions. A function $\phi:[0,\infty) \to [0,\infty)$ is a \textit{gauge function} if $\phi$ is monotonically increasing, strictly increasing for $t > 0$ and continuous. If $\phi$ is a gauge, define the $phi$-Hausdorff content at precision $r$ by
\begin{center}
$h^\phi_r(E) = \inf\left\{ \sum_i \phi(d(Q_i)) \, |\, \bigcup_i Q_i \text{ covers } E \text{ and } d(Q_i) \leq 2^{-r}\right\}$,
\end{center}
where $d(Q)$ is the diameter of ball $Q$. We define the $phi$-Hausdorff measure of $E$ by
\begin{center}
$\mathcal{H}^\phi(E) = \lim\limits_{r\to \infty} h^\phi_r(E)$.
\end{center}
Thus we recover the $s$-dimensional Hausdorff measure when $\phi(t) = t^s$. 

Gauged Hausdorff measures give fine-grained information about the size of a set. There are sets $E$ which Hausdorff dimension $s$, but $\Hs(E) = 0$ or $\Hs(E) = \infty$. However, it is sometimes possible to find an appropriate gauge so that $0 < \mathcal{H}^\phi(E) < \infty$. When $0 < \mathcal{H}^\phi(E) <\infty$, we say that $\phi$ is an \textit{exact gauge for $E$}.
\begin{ex}
For almost every Brownian path $X$ in $\R^2$, $\mathcal{H}^2(X) = 0$, but $0<\mathcal{H}^\phi(X) <\infty$, where
$\phi(t) = t^2 \log\frac{1}{t} \log\log\frac{1}{t}$.
\end{ex}

For two outer measures $\mu$ and $\nu$, $\mu$ is said to be \textit{absolutely continuous with respect to} $\nu$, denoted $\mu \ll \nu$, if $\mu(A) = 0$ for every set $A$ for which $\nu(A) = 0$.
\begin{ex}
For every $s$, let $\phi_s(t) = t^s \log \frac{1}{t}$. Then $\Hs \ll \mathcal{H}^{\phi_s}$.
\end{ex}

\begin{ex}
For every $s$, let $\phi_s(t) = \frac{t^s}{\log \frac{1}{t}}$. Then $\mathcal{H}^{\phi_s} \ll  \Hs$.
\end{ex}

\subsection{Algorithmic Information Theory}\label{ssec:ait}

The \emph{conditional Kolmogorov complexity} of a binary string $\sigma\in\{0,1\}^*$ given binary string $\tau\in\{0,1\}^*$ is 
		\[K(\sigma|\tau)=\min_{\pi\in\{0,1\}^*}\left\{\ell(\pi):U(\pi,\tau)=\sigma\right\}\,,\]
		where $U$ is a fixed universal prefix-free Turing machine and $\ell(\pi)$ is the length of $\pi$. The \emph{Kolmogorov complexity} of $\sigma$ is $K(\sigma)=K(\sigma|\lambda)$, where $\lambda$ is the empty string. An important fact is that the choice of universal machine affects the Kolmogorov complexity by at most an additive constant (which, especially for our purposes, can be safely ignored). See \cite{LiVit08,Nies09,DowHir10} for a more comprehensive overview of Kolmogorov complexity.

We can naturally extend these definitions to Euclidean spaces by introducing ``precision" parameters~\cite{LutMay08,LutLut17}. Let $x\in\R^m$, and $r,s\in\N$. The \emph{Kolmogorov complexity of $x$ at precision $r$} is
\[K_r(x)=\min\left\{K(p)\,:\,p\in B_{2^{-r}}(x)\cap\Q^m\right\}\,.\]
The \emph{conditional Kolmogorov complexity of $x$ at precision $r$ given $q\in\Q^m$} is
\[\hat{K}_r(x|q)=\min\left\{K(p\,|\,q)\,:\,p\in B_{2^{-r}}(x)\cap\Q^m\right\}\,.\]
The \emph{conditional Kolmogorov complexity of $x$ at precision $r$ given $y\in\R^n$ at precision $s$} is
\[K_{r,s}(x|y)=\max\big\{\hat{K}_r(x|q)\,:\,q\in B_{2^{-s}}(y)\cap\Q^n\big\}\,.\]
We typically abbreviate $K_{r,r}(x|y)$ by $K_r(x|y)$.

		The \emph{effective Hausdorff dimension} and \emph{effective packing dimension}\footnote{Although effective Hausdorff was originally defined by J. Lutz~\cite{Lutz03b} using martingales, it was later shown by Mayordomo~\cite{Mayordomo02} that the definition used here is equivalent. For more details on the history of connections between Hausdorff dimension and Kolmogorov complexity, see~\cite{DowHir10,Mayordomo08}.} of a point $x\in\R^n$ are
	\[\dim(x)=\liminf_{r\to\infty}\frac{K_r(x)}{r}\quad\text{and}\quad\Dim(x) = \limsup_{r\to\infty}\frac{K_r(x)}{r}\,.\]

	By letting the underlying fixed prefix-free Turing machine $U$ be a universal \emph{oracle} machine, 
	we may \emph{relativize} the definition in this section to an arbitrary oracle set $A \subseteq \N$. The definitions of $K^A_r(x)$, $\dim^A(x)$, $\Dim^A(x)$, etc. are then all identical to their unrelativized versions, except that $U$ is given oracle access to $A$. Note that taking oracles as subsets of the naturals is quite general. We can, and frequently do, encode a point $y$ into an oracle, and consider the complexity of a point \textit{relative} to $y$. In these cases, we typically forgo explicitly referring to this encoding, and write e.g. $K^y_r(x)$. We can also \textit{join} two oracles $A, B\subseteq\N$ using any computable bijection $f: \N \times \N \to \N$. We denote the join of $A$ and $B$ by $(A,B)$. We can generalize this procedure to join any countable sequence of oracles.
	
As mentioned in the introduction, the connection between effective dimensions and the classical Hausdorff and packing dimensions is given by the point-to-set principle introduced by J. Lutz and N. Lutz \cite{LutLut17}.
\begin{thm}[Point-to-set principle]\label{thm:p2s}
Let $n \in \N$ and $E \subseteq \R^n$. Then
\begin{align*}
\dim_H(E) &= \min\limits_{A \subseteq \N} \sup_{x \in E} \dim^A(x), \text{ and}\\
\dim_P(E) &= \min\limits_{A \subseteq \N} \sup_{x \in E} \Dim^A(x)\,.
\end{align*}
\end{thm}
An oracle testifying to the the first equality is called a \textit{Hausdorff oracle for E}. Similarly, an oracle testifying to the the second equality is called a \textit{packing oracle for E}.

\section{Optimal Hausdorff Oracles}\label{sec:optimaloracles}
For any set $E$, there are infinitely many Hausdorff oracles for $E$. A natural question is whether there is a Hausdorff oracle which minimizes the complexity of every point in $E$. Unfortunately, it is, in general, not possible for a single oracle to maximally reduce \textit{every} point. We introduce the notion of optimal Hausdorff oracles by weakening the condition to a \textit{single} point. 
\begin{defn}
Let $E \subseteq\R^n$ and $A \subseteq \N$. We say that $A$ is \textit{Hausdorff optimal} for $E$ if the following conditions are satisfied.
\begin{enumerate}
\setlength\itemsep{.5em}
\item $A$ is a Hausdorff oracle for $E$.
\item For every $B \subseteq \N$ and every $\epsilon > 0$ there is a point $x\in E$ such that $\dim^{A,B}(x) \geq \dim_H(E) - \epsilon$ and for almost every $r\in \N$
\begin{center}
$K^{A,B}_r(x) \geq K^A_r(x) - \epsilon r$.
\end{center}
\end{enumerate}
\end{defn}
Note that the second condition only guarantees the existence of \textit{one} point whose complexity is unaffected by the addtional information in $B$. However, we can show that this implies the seemingly stronger condition that ``most" points are unaffected. For $B\subseteq \N$, $\epsilon > 0$ define the set
\begin{center}
$N(A, B,\epsilon) = \{x \in E \, | \, (\forall^\infty r) \, K^{A,B}_r(x) \geq K^A_r(x) - \epsilon r\}$
\end{center}

\begin{prop}
Let $E\subseteq\R^n$ be a set such that $\dim_H(E) > 0$ and let $A$ be an oracle. Then $A$ is a Hausdorff optimal oracle for $E$ if and only if $A$ is a Hausdorff oracle and $\dim_H(N(A, B,\epsilon)) = \dim_H(E)$ for every $B\subseteq\N$ and $\epsilon > 0$.
\end{prop}
\begin{proof}
For the forward direction, let $A$ be a optimal Hausdorff oracle for $E$. Then by the first condition of the definition, $A$ is a Hausdorff oracle. Let $B\subseteq\N$ and $\epsilon > 0$. Let $C$ be a Hausdorff oracle for $N(A, B,\epsilon)$. For the sake of contradiction, suppose that 
\begin{center}
$\dim_H(N(A, B,\epsilon)) < \dim_H(E) - \gamma$,
\end{center}
for some $\gamma > 0$. We will, without loss of generality, assume that $\gamma < \epsilon$. Then, by the point to set principle, for every $x \in N(A,B,\epsilon)$,
\begin{align*}
\dim^{A, (B, C)}(x) &\leq \dim^{C}(x)\\
&\leq \dim_H(N(A,B,\epsilon))\\
&< \dim_H(E) - \gamma.
\end{align*}
Since, $A$ is an optimal Hausdorff oracle for $E$, there is a point $x \in E$ such that $\dim^{A, (B,C)}(x) \geq \dim_H(E) - \gamma$ and for almost every $r\in \N$
\begin{center}
$K^{A,(B, C)}_r(x) \geq K^A_r(x) - \gamma r$.
\end{center}
By our previous discussion, any such point $x$ cannot be in $N(A,B,\epsilon)$. However, if $x \notin N(A,B,\epsilon)$, then for infinitely many $r$,
\begin{center}
$K^{A,(B, C)}_r(x) < K^A_r(x) - \epsilon r$.
\end{center}
Thus, no such $x$ exists, contradicting the fact that $A$ is Hausdorff optimal.

For the backward direction, let $A$ be an oracle satisfying the hypothesis. Then $A$ is a Hausdorff oracle for $E$ and the first condition of optimal Hausdorff oracles is satisfied. Let $B \subseteq \N$ and $\epsilon > 0$. By our hypothesis and the point-to-set principle,
\begin{align*}
\dim_H(E) &= \dim_H(N(A, B,\epsilon))\\
&\leq \sup\limits_{x \in N(A, B,\epsilon)} \dim^{A,B}(x).
\end{align*}
Therefore, there is certainly a point $x \in E$ such that $\dim^{A,B}(x) \geq \dim_H(E) - \epsilon$ and 
\begin{center}
$K^{A,B}_r(x) \geq K^A_r(x) - \epsilon r$,
\end{center}
for almost every $r\in\N$.
\end{proof}

A simple, but useful, result is if $B$ is an oracle obtained by adding additional information to an optimal Hausdorff oracle, then $B$ is also optimal.
\begin{lem}\label{lem:joinOptimalisOptimal}
Let $E\subseteq \R^n$. If $A$ is an optimal Hausdorff oracle for $E$, then the join $C = (A,B)$ is Hausdorff optimal for $E$ for every oracle $B$. 
\end{lem}
\begin{proof}
Let $A$ be an optimal Hausdorff oracle for $E$. By the point-to-set principle (Theorem \ref{thm:p2s}),
\begin{align*}
\dim_H(E) &= \sup\limits_{x\in E} \dim^A(x)\\
&\geq \sup\limits_{x\in E} \dim^{C}(x)\\
&\geq \dim_H(E).
\end{align*}
Hence, the oracle $C = (A,B)$ is a Hausdorff oracle for $E$.

Let $B^\prime \subseteq \N$ be an oracle, and let $\epsilon > 0$. Let $x \in E$ be a point such that
\begin{equation}
\dim^{A, (B, B^\prime)}(x) \geq \dim_H(E) - \epsilon / 2,
\end{equation}
and
\begin{equation}
K_r^{A, (B, B^\prime)}(x) \geq K^A_r(x) - \epsilon r / 2,
\end{equation}
for almost every precision $r$. Note that such a point exists since $A$ is optimal for $E$.

For all sufficiently large $r$,
\begin{align*}
K^{(A, B), B^\prime}_r(x) &= K^{A, (B, B^\prime)}_r(x)\\
&\geq K^{A}_r(x) - \epsilon r/2\\
&\geq K^{A, B}_r(x) - \epsilon r/2\\
&= K^{C}_r(x) - \epsilon r/2.
\end{align*}
Therefore, $C = (A,B)$ is optimal for $E$.

\end{proof}

We now give some basic closure properties of the class of sets with optimal Hausdorff oracles. 
\begin{obs}\label{obs:optimalSubset}
Let $F \subseteq E$. If $\dim_H(F) = \dim_H(E)$ and $F$ has an optimal Hausdorff oracle, then $E$ has an optimal Hausdorff oracle.

\end{obs}

We can also show that having optimal Hausdorff oracles is closed under countable unions.
\begin{prop}\label{prop:optimalOraclesClosedUnderUnion}
Let $E_1,E_2,\ldots$ be a countable sequence of sets and let $E = \cup_n E_n$. If every set $E_n$ has an optimal Hausdorff oracle, then $E$ has an optimal Hausdorff oracle.
\end{prop}
\begin{proof}
We first note that 
\begin{center}
$\dim_H(E) = \sup_n \dim_H(E_n)$.
\end{center}
For each $n$, let $A_n$ be an optimal Hausdorff oracle for $E_n$. Let $A$ be the join of $A_1, A_2,\ldots$. Let $B$ be a Hausdorff oracle for $E$. Note that, by Lemma \ref{lem:joinOptimalisOptimal}, for every $n$, since $A_n$ is an optimal Hausdorff oracle for $E_n$, $(A,B)$ is optimal for $E_n$.

We now claim that $(A, B)$ is an optimal Hausdorff oracle for $E$. Theorem \ref{thm:p2s} shows that item (1) of the definition of optimal Hausdorff oracles is satisfied. For item (2), let $C \subseteq \N$ be an oracle, and let $\epsilon > 0$. Let $n$ be a number such that $\dim_H(E_n) > \dim_H(E) - \epsilon$. Since $(A,B)$ is Hausdorff optimal for $E_N$, there is a point $x \in E_n$ such that
\begin{enumerate}
\item[(i)] $\dim^{(A, B), C}(x) \geq \dim_H(E_n) - \epsilon \geq \dim_H(E) - \epsilon$, and
\item[(ii)] for almost every $r$,
\begin{center}
$K^{(A, B), C}_r(x) \geq K^{(A,B)}_r(x) - \epsilon r$.
\end{center}
\end{enumerate}
Therefore, item (2) of the definition of optimal Hausdorff oracles is satisfied, and so $(A,B)$ is Hausdorff optimal for $E$.

\end{proof}

\subsection{Outer Measures and Optimal Oracles}\label{ssec:optoraclesouter}

In this section we give a sufficient condition for a set to have optimal Hausdorff oracles. Specifically, we prove that if $\dim_H(E) = s$, and there is a metric outer measure, absolutely continuous with respect to $\Hs$, such that $0 < \mu(E) < \infty$, then $E$ has optimal Hausdorff oracles. Although stated in this general form, the main application of this result (in Section \ref{ssec:setswithopt}) is for the case $\mu = \Hs$.

For every $r\in\N$, let $\mathcal{Q}^n_r$ be the set of all dyadic cubes at precision $r$, i.e., cubes of the form
\begin{center}
$Q = [m_1 2^{-r}, (m_1 + 1) 2^{-r}) \times \ldots \times [m_n 2^{-r}, (m_n + 1) 2^{-r})$,
\end{center}
where $0\leq m_1,\ldots, m_n \leq 2^{r}$. For each $r$, we refer to the $2^{nr}$ cubes in $\mathcal{Q}_r$ as $Q_{r,1}, \ldots, Q_{r,2^{nr}}$. We can identify each dyadic cube $Q_{r,i}$ with the unique dyadic rational $d_{r,i}$ at the center of $Q_{r,i}$.

We now associate, to each metric outer measure, a \textit{discrete semimeasure on the dyadic rationals $\mathbb{D}$}. Recall that discrete semimeasure on $\mathbb{D}^n$ is a function $p: \mathbb{D}^n \to [0,1]$ which satisfies $\Sigma_{r,i} p(d_{r,i}) < \infty$.

Let $E\subseteq\R^n$ and $\mu$ be a metric outer measure such that $0<\mu(E) <\infty$. Define the function $p_\mu : \D^n \rightarrow [0,1]$ by
\begin{center}
$p_{\mu,E}(d_{r,i}) = \frac{\mu(E \cap Q_{r,i})}{r^2\mu(E)}$.
\end{center}

\begin{obs}\label{obs:dyadicSum}
Let $\mu$ be a metric outer measure and $E\subseteq\R^n$ such that $0 < \mu(E) < \infty$. Then for every $r$, every dyadic cube $Q \in \mathcal{Q}_r$, and all $r^\prime > r$,
\begin{center}
$\mu(E \cap Q) = \sum\limits_{\substack{Q^\prime \subset Q\\ Q^\prime \in \mathcal{Q}_{r^\prime}}} \mu(E \cap Q^\prime)$.
\end{center}
\end{obs}

\begin{prop}
Let $E\subseteq\R^n$ and $\mu$ be a metric outer measure such that $0 < \mu(E) < \infty$. Relative to some oracle $A$, the function $p_{\mu,E}$ is a lower semi-computable discrete semimeasure.
\end{prop}
\begin{proof}
We can encode the real numbers $p_{\mu,E}(d)$ into an oracle $A$, relative to which $p_{\mu,E}$ is clearly computable.

To see that $p_{\mu, E}$ is indeed a discrete semimeasure, by Observation \ref{obs:dyadicSum},
\begin{align*}
\sum\limits_{r,i} p_{\mu,E}(d_{r,i}) &= \sum\limits_{r} \sum\limits_{i=1}^{2^{2r}} \frac{\mu(E \cap Q_{r,i})}{r^2\mu(E)}\\
&= \sum\limits_{r} \frac{1}{r^2\mu(E)}\sum\limits_{i=1}^{2^{2r}} \mu(E \cap Q_{r,i})\\
&= \sum\limits_{r} \frac{\mu(E)}{r^2\mu(E)}\\
&< \infty.
\end{align*}
\end{proof}

In order to connect the existence of such an outer measure $\mu$ to the existence of optimal oracles, we need to relate the semimeasure $p_\mu$ and Kolmogorov complexity. We achieve this using a fundamental result in algorithmic information theory.

Levin's optimal lower semicomputable subprobability measure, relative to an oracle $A$, on the dyadic rationals $\mathbb{D}$ is defined by
\begin{center}
$\mathbf{m}^A(d) = \sum\limits_{\pi \, : \, U^A(\pi) = d} 2^{-|\pi|}$.
\end{center}

\begin{lem}\label{lem:LevinOptimal}
Let $E\subseteq\R^n$ and $\mu$ be a metric outer measure such that $0 < \mu(E) < \infty$. Let $A$ be an oracle relative to which $p_{\mu, E}$ is lower semi-computable. Then is a constant $\alpha > 0$ such that $\mathbf{m}^A(d) \geq  \alpha p_{\mu,E}(d)$, for every $d \in \D^n$.
\end{lem}
\begin{proof}
Case and Lutz \cite{CasLut15}, generalizing Levin's coding theorem \cite{Levin73, Levin74}, showed that there is a constant $c$ such that
\begin{center}
$\mathbf{m}^A(d_{r,i}) \leq 2^{-K^A(d_{r,i}) + K^A(r) + c}$,
\end{center}
for every $r\in\N$ and $d_{r,i} \in \mathbb{D}^n$. The optimality of $\mathbf{m}^A$ ensures that, for every lower semicomputable (relative to $A$) discrete semimeasure $\nu$ on $\mathbb{D}^n$,
\begin{center}
$\mathbf{m}^A(d_{r,i}) \geq \alpha \nu(d_{r,i})$.
\end{center}
\end{proof}

The results of this section have dealt with the dyadic rationals. However, we ultimately deal with the Kolmogorov complexity of Euclidean points. A result of Case and Lutz \cite{CasLut15} relates the Kolmogorov complexity of Euclidean points with the complexity of dyadic rationals.
\begin{lem}[\cite{CasLut15}]\label{lem:CasLutDyadic}
Let $x \in [0,1)^n$, $A \subseteq \N$, and $r\in \N$. Let $Q_{r,i}$ be the (unique) dyadic cube at precision $r$ containing $x$. Then
\begin{center}
$K^A_r(x) = K^A(d_{r,i}) - O(\log r)$.
\end{center}
\end{lem}

\begin{lem}\label{lem:mainTheoremEngine}
Let $E\subseteq\R^n$ and $\mu$ be a metric outer measure such that $0 < \mu(E) < \infty$. Let $A$ be an oracle relative to which $p_{\mu, E}$ is lower semi-computable. Then, for every oracle $B \subseteq\N$ and every $\epsilon > 0$, the set
\begin{center}
$N = \{x \in E \, | \, (\exists^\infty) \; K^{A,B}_r(x) < K^A_r(x) - \epsilon r\}$
\end{center}
has $\mu$-measure zero.
\end{lem}
\begin{proof}
Let $B\subseteq\N$ and $\epsilon > 0$. For every $R\in\N$, there is a set $\mathcal{C}_R$ of dyadic cubes satisfying the following.
\begin{itemize}
\item The cubes in $\mathcal{C}_R$ cover $N$.
\item Every $Q_{r,i}$ in $\mathcal{C}_R$ satisfies $r\geq R$.
\item For every $Q_{r,i} \in \mathcal{C}_R$, 
\begin{center}
$K^{A,B}(d_{r,i})< K^A(d_{r,i}) - \epsilon r + O(\log r)$.
\end{center}
\end{itemize}
Note that the last item follows from our definition of $N$ by Lemma \ref{lem:CasLutDyadic}.

Since the family of cubes in $\mathcal{C}_R$ covers $N$, by the subadditive property of $\mu$, 
\begin{center}
$\sum\limits_{Q_{r,i} \in \mathcal{C}_R} \mu(E\cap Q_{r,i}) \geq \mu(N)$.
\end{center}
Thus, for every $R$, by Lemma \ref{lem:LevinOptimal} and Kraft's inequality,
\begin{align*}
1 &\geq \sum\limits_{Q_{r,i} \in \mathcal{C}_R} 2^{-K^{A,B}(d_{r,i})} \\ 
&\geq \sum\limits_{Q_{r,i} \in \mathcal{C}_R} 2^{\epsilon r - K^{A}(d_{r,i})}\\
&\geq \sum\limits_{Q_{r,i} \in \mathcal{C}_R} 2^{\epsilon r} \mathbf{m}^A(d_{r,i})\\
&\geq \sum\limits_{Q_{r,i} \in \mathcal{C}_R} 2^{\epsilon r - K^A(r) + c}\alpha p_{\mu,E}(d_{r,i})\\
&\geq \sum\limits_{Q_{r,i} \in \mathcal{C}_R} 2^{\epsilon r/2 } p_{\mu,E}(d_{r,i})\\
&\geq \sum\limits_{Q_{r,i} \in \mathcal{C}_R} 2^{\epsilon r/2} \frac{\mu(E \cap Q_{r,i})}{r^2\mu(E)}\\
&\geq \sum\limits_{Q_{r,i} \in \mathcal{C}_R} 2^{\epsilon r/4} \frac{\mu(E \cap Q_{r,i})}{r^2\mu(E)}\\
&\geq 2^{\epsilon R/4} \sum\limits_{Q_{r,i} \in \mathcal{C}_R} \frac{\mu(E\cap Q_{r,i})}{\mu(E)}\\
&\geq 2^{\epsilon R/4} \frac{\mu(N)}{\mu(E)}.
\end{align*}

Since $R$ can be arbitrarily large, and $0< \mu(E) < \infty$, the conclusion follows.
\end{proof}

We now have the machinery in place to prove the main theorem of this section. 
\begin{thm}\label{thm:mainTheorem}
Let $E\subseteq\R^n$ with $\dim_H(E) = s$. Suppose there is a metric outer measure $\mu$ such that
\begin{center}
$0<\mu(E)<\infty$,
\end{center}  
and either 
\begin{enumerate}
\item $\mu \ll \mathcal{H}^{s-\delta}$, for every $\delta > 0$, or
\item $\Hs \ll \mu$ and $\Hs(E) > 0$.
\end{enumerate}
Then $E$ has an optimal Hausdorff oracle $A$.
\end{thm}
\begin{proof}
Let $A\subseteq\N$ be a Hausdorff oracle for $E$ such that $p_{\mu,E}$ is computable relative to $A$. Note that such an oracle exists by the point-to-set principle and routine encoding. We will show that $A$ is optimal for $E$. 

For the sake of contradiction, suppose that there is an oracle $B$ and $\epsilon > 0$ such that, for every $x \in E$ either
\begin{enumerate}
\item $\dim^{A,B}(x) < s - \epsilon$, or
\item there are infinitely many $r$ such that $K^{A,B}_r(x) < K^A_r(x) - \epsilon r$.
\end{enumerate}

Let $N$ be the set of all $x$ for which the second item holds. By Lemma \ref{lem:mainTheoremEngine}, $\mu(N) = 0$. We also note that, by the point-to-set principle, 
\begin{center}
$\dim_H(E - N) \leq s - \epsilon$,
\end{center}
and so $\Hs(E - N) = 0$.

To achieve the desired contradiction, we first assume that $\mu \ll \mathcal{H}^{s-\delta}$, for every $\delta > 0$. Since $\mu \ll \mathcal{H}^{s-\delta}$, and $\dim_H(E-N) < s-\epsilon$,
\begin{center}
$\mu(E - N) = 0$.
\end{center}
Since $\mu$ is a metric outer measure, 
\begin{align*}
0 &< \mu(E)\\
&\leq \mu(N) + \mu(E-N)\\
&= 0,
\end{align*}
a contradiction. 

Now suppose that $\Hs \ll \mu$ and $\Hs(E) > 0$. Then, since $\Hs$ is an outer measure, $\Hs(E) > 0$ and $\Hs(E - N) = 0$ we must have $\Hs(N) >0$. However this implies that  $\mu(N) > 0$, and we again have the desired contradiction. Thus $A$ is an optimal Hausdorff oracle for $E$ and the proof is complete.
\end{proof}

Recall that $E \subseteq [0,1)^n$ is called an $s$-set if
\begin{center}
$0 < \Hs(E) <\infty$.
\end{center}
Since $\Hs$ is a metric outer measure, and trivially absolutely continuous with respect to itself, we have the following corollary.
\begin{cor}
Let $E \subseteq [0,1)^n$ be an $s$-set. Then there is an optimal Hausdorff oracle for $E$.
\end{cor}

\subsection{Sets with optimal Hausdorff oracles}\label{ssec:setswithopt}
We now show that every analytic set has optimal Hausdorff oracles. 
\begin{lem}\label{lem:analyticHasOptimal}
Every analytic set $E$ has optimal Hausdorff oracles.
\end{lem}
\begin{proof}
We begin by assuming that $E$ is compact, and let $s = \dim_H(E)$. Then for every $t < s$, $\mathcal{H}^t(E) > 0$. Thus, by Theorem \ref{thm:compactSset}(1), there is a sequence of compact subsets $F_1,F_2,\ldots$ of $E$ such that 
\begin{center}
$\dim_H(\bigcup_n F_n) = \dim_H(E)$,
\end{center}
and, for each $n$,
\begin{center}
$0 < \mathcal{H}^{s_n}(F_n) < \infty$,
\end{center}
where $s_n = s - 1/n$. Therefore, by Theorem \ref{thm:mainTheorem}, each set $F_n$ has optimal Hausdorff oracles. Hence, by Proposition \ref{prop:optimalOraclesClosedUnderUnion}, $E$ has optimal Hausdorff oracles and the conclusion follows.

We now show that every $\Sigma^0_2$ set has optimal Hausdorff oracles. Suppose $E = \cup_n F_n$ is $\Sigma^0_1$, where each $F_n$ is compact. As we have just seen, each $F_n$ has optimal Hausdorff oracles. Therefore, by Proposition \ref{prop:optimalOraclesClosedUnderUnion}, $E$ has optimal Hausdorff oracles and the conclusion follows. 

Finally, let $E$ be analytic. By Theorem \ref{thm:compactSset}(2), there is a $\Sigma^0_2$ subset $F$ of the same Hausdorff dimension as $E$. We have just seen that $F$ must have an optimal Hausdorff oracle. Since $\dim_H(F) = \dim_H(E)$, by Observation \ref{obs:optimalSubset} $E$ has optimal Hausdorff oracles, and the proof is complete
\end{proof}

Crone, Fishman and Jackson \cite{CroFisJac20} have recently shown that, assuming the Axiom of Determinacy (AD)\footnote{Note that AD is inconsistent with the axiom of choice.}, \textit{every} subset $E$ has a Borel subset $F$ such that $\dim_H(F) = \dim_H(E)$. This, combined with Lemma \ref{lem:analyticHasOptimal}, yields the following corollary. 
\begin{cor}
Assuming AD, every set $E\subseteq \R^n$ has optimal Hausdorff oracles.
\end{cor}

\begin{lem}\label{lem:regularHasOptimalHaus}
Suppose that $E\subseteq\R^n$ satisfies $\dim_H(E) = \dim_P(E)$. Then $E$ has an optimal Hausdorff oracle. Moreover, the join $(A,B)$ is an optimal Hausdorff oracle, where $A$ and $B$ are the Hausdorff and packing oracles, respectively, of $E$.
\end{lem}
\begin{proof}
Let $A$ be a Hausdorff oracle for $E$ and let $B$ be a packing oracle for $E$. We claim that that the join $(A, B)$ is an optimal Hausdorff oracle for $E$. By the point-to-set principle, and the fact that extra information cannot increase effective dimension,
\begin{align*}
\dim_H(E) &= \sup\limits_{x\in E} \dim^{A}(x)\\
&\geq \sup\limits_{x\in E} \dim^{A,B}(x)\\
&\geq \dim_H(E).
\end{align*} 
Therefore  
\begin{center}
$\dim_H(E) = \sup\limits_{x\in E} \dim^{A,B}(x)$,
\end{center}
and the first condition of optimal Hausdorff oracles is satisfied.

Let $C\subseteq\N$ be an oracle and $\epsilon > 0$. By the point-to-set principle, 
\begin{center}
$\dim_H(E) \leq \sup\limits_{x\in E} \dim^{A,B,C}(x)$,
\end{center}
so there is an $x\in E$ such that 
\begin{center}
$\dim_H(E) - \epsilon /4 < \dim^{A,B,C}(x)$.
\end{center}

Let $r$ be sufficiently large. Then, by our choice of $B$ and the fact that additional information cannot increase the complexity of a point,
\begin{align*}
K^{A,B}_r(x) &\leq K^{B}_r(x)\\
&\leq \dim_P(E)r + \epsilon r/ 4\\
&= \dim_H(E)r + \epsilon r /4 \\
&< \dim^{A,B,C}(x)r + \epsilon r/ 2\\
&\leq K_r^{A,B,C}(x) + \epsilon r.
\end{align*}

Since the oracle $C$ and $\epsilon$ were arbitrarily, the proof is complete.
\end{proof}

\subsection{Sets without optimal Hausdorff oracles}\label{ssec:counterexamOptimal}
In the previous section, we gave general conditions for a set $E$ to have optimal Hausdorff oracles. Indeed, we saw that under the axiom of determinacy, every set has optimal Hausdorff oracles.

However, assuming the axiom of choice (AC) and the continuum hypothesis (CH), we are able to construct sets without optimal Hausdorff oracles. 

\begin{lem}\label{lem:setRwithoutOptimalOracles}
Assume AC and CH. Then, for every $s\in (0,1)$, there is a subset $E\subseteq\R$ with $\dim_H(E) = s$ such that $E$ does not have optimal Hausdorff oracles.
\end{lem}

Let $s \in (0,1)$. We begin by defining two sequences of natural numbers, $\{a_n\}$ and $\{b_n\}$. Let $a_1 = 2$, and $b_1 = \lfloor 2/s\rfloor$. Inductively define $a_{n+1} = b_{n}^2$ and $b_{n+1} = \lfloor a_{n+1}/s\rfloor$. Note that 
\begin{center}
$\lim_n a_n / b_n = s$.
\end{center}

Using AC and CH, we order the subsets of the natural numbers such that every subset has countably many predecessors. For every countable ordinal $\alpha$, let $f_\alpha : \N \to \{\beta \, | \, \beta < \alpha\}$ be a function such that each ordinal $\beta$ strictly less than $\alpha$ is mapped to by infinitely many $n$. Note that such a function exists, since the range is countable assuming CH.

We will define real numbers $x_\alpha$, $y_\alpha$ via transfinite induction. Let $x_1$ be a real which is random relative to $A_1$. Let $y_1$ be the real whose binary expansion is given by
\begin{align*}
y_1[r] = \begin{cases}
0 &\text{ if } a_n < r \leq b_n \text{ for some } n \in \N \\
x_1[r] &\text{ otherwise}
\end{cases}
\end{align*}

For the induction step, suppose we have defined our points up to $\alpha$. Let $x_\alpha$ be a real number which is random relative to the join of $\bigcup_{\beta < \alpha} (A_\beta, x_\beta)$ and $A_\alpha$. This is possible, as we are assuming that this union is countable. Let $y_\alpha$ be the point whose binary expansion is given by
\begin{align*}
y_\alpha[r] = \begin{cases}
x_\beta[r] &\text{ if } a_n < r \leq b_n, \text{ where } f_\alpha(n) = \beta \\
x_\alpha[r] &\text{ otherwise}
\end{cases}
\end{align*}

Finally, we define our set $E = \{y_\alpha\}$. We now claim that $\dim_H(E) = s$, and that $E$ does not have an optimal Hausdorff oracle. 

\begin{lem}
The Hausdorff dimension of $E$ is $s$.
\end{lem}
\begin{proof}
We first upper bound the dimension. Let $A$ be an oracle encoding $x_1$. From our construction, for every element $y \in E$, there are infinitely many intervals $[a_n, b_n]$ such that $y[a_n,b_n] = x_1[a_n, b_n]$. Hence, for every $y \in E$, there are infinitely many $n$ such that
\begin{align*}
K^{A}_{b_n}(y) &= K^{A}_{a_n}(y) + K^{A}_{b_n, a_n}(y) + o(b_n)\\
&\leq K^{A}_{a_n}(y)+ o(b_n)\\
&\leq a_n+ o(b_n).
\end{align*}

Therefore, by the point to set principle,
\begin{align*}
\dim_H(E) &\leq \sup_y \dim^A(y)\\
&= \sup_y \liminf_r \frac{K^A_r(y)}{r}\\
&\leq \sup_y \liminf_n \frac{K^{A}_{b_n}(y)}{b_n}\\
&\leq \sup_y \liminf_n \frac{a_n + o(b_n)}{b_n}\\
&\leq \sup_y \liminf_n s\\
&= s,
\end{align*}
and the proof that $\dim_H(E) \leq s$ is complete.

For the lower bound, let $A$ be a Hausdorff oracle for $E$, and let $\alpha$ be the ordinal corresponding to $A$. By our construction of $y_\alpha$, for every $n$, 
\begin{align*}
K^{A_\alpha}_{a_n}(y_\alpha) &\geq K^{A_\alpha}_{a_n}(x_\alpha) - b_{n-1}\\
&\geq a_n - b_{n-1} - o(a_n)\\
&\geq a_n - a_{n}^{\frac{1}{2}} - o(a_n).
\end{align*}
Hence, for every $n$, and every $a_n < r \leq b_n$,
\begin{align*}
K^{A_\alpha}_{r}(y_\alpha) &\geq K^{A_\alpha}_{a_n}(y_\alpha)\\
&\geq a_n - a_{n}^{\frac{1}{2}} - o(a_n).
\end{align*}
This implies that
\begin{center}
$\frac{K^{A_\alpha}_{r}(y_\alpha)}{r} \geq s - o(1)$,
\end{center}
for every $n$, and every $a_n < r \leq b_n$.

We can also conclude that, for every $n$ and every $b_n < r \leq a_{n+1}$,
\begin{align*}
K^{A_\alpha}_{r}(y_\alpha) &= K^{A_\alpha}_{b_n}(y_\alpha) + K^{A_\alpha}_{r,b_n}(y_\alpha) - o(r)\\
&\geq a_n - a_{n}^{\frac{1}{2}} + r - b_n - o(r).
\end{align*}
This implies that
\begin{align*}
\frac{K^{A_\alpha}_{r}(y_\alpha)}{r} &= 1 + \frac{a_n}{r} - \frac{b_n}{r} - o(1)\\
&= 1 - \frac{a_n(1/s - 1)}{r} - o(1)\\
&\geq 1 - s(1/s - 1) - o(1)\\
&= s - o(1).
\end{align*}
for every $n$, and every $a_n < r \leq b_n$.

Together, these inequalities and the point-to-set principle show that
\begin{align*}
\dim_H(E) &= \sup_x \dim^A(x)\\
&\geq \dim^A(y_\alpha)\\
&= \liminf_{r} \frac{K^A(y_\alpha)}{r}\\
&\geq \liminf_{r} s - o(1)\\
&= s,
\end{align*}
and the proof is complete.
\end{proof}

\begin{lem}
$E$ does not have optimal Hausdorff oracles.
\end{lem}
\begin{proof}
Let $A_\alpha \subseteq \N$ be an oracle. It suffices to show that $A_\alpha$ is not optimal. With this goal in mind, let $B$ be an oracle encoding $x_\alpha$ and the set $\{y_\beta \, | \, \beta < \alpha\}$. Note that we can encode this information since this set is countable. 

Let $y_\beta \in E$. First, suppose that $\beta \leq \alpha$. Then by our choice of $B$, $\dim^{A_\alpha, B}(y_\beta) = 0$. So then suppose that $\beta > \alpha$. We first note that, since $x_\beta$ is random relative to $A_\alpha$
\begin{align*}
K^{A_\alpha}_{a_n}(y_\beta) &\geq K^{A_\alpha}(y_\beta[b_{n-1}\ldots a_n]) - O(\log a_n)\\
&= K^{A_\alpha}(x_\beta[b_{n-1}\ldots a_n]) - O(\log a_n)\\
&\geq a_n - b_{n-1} - O(\log a_n)\\
&\geq a_n - o(a_n),
\end{align*}
for every $n \in \N$.

By our construction, there are infinitely many $n$ such that
\begin{equation}\label{eq:ybetaEqualsXAlpha}
y_\beta[a_n\ldots b_n] = x_\alpha[a_n\ldots b_n]
\end{equation}
Since $x_\alpha$ is random relative to $A_\alpha$, for any $n$ such that (\ref{eq:ybetaEqualsXAlpha}) holds,
\begin{align*}
K^{A_\alpha}_{b_n}(y_\beta) &=  K^{A_\alpha}_{a_n}(y_\beta) + K^{A_\alpha}_{b_n,a_n}(y_\beta) \\
&\geq a_n - o(a_n) + K^{A_\alpha}(y_\beta[a_{n}\ldots b_n]) - O(\log b_n)\\
&= a_n - o(a_n) + K^{A_\alpha}(x_\alpha[a_{n}\ldots b_n]) \\
&\geq a_n - o(a_n) + b_n - a_n - o(b_n) \\
&= b_n - o(b_n).
\end{align*}

However, since we can compute $x_\alpha$ given $B$, 
\begin{align*}
K^{A_\alpha,B}_{b_n}(y_\beta) &= K^{A_\alpha, B}_{a_n}(y_\beta) + K^{A_\alpha, B}_{b_n,a_n}(y_\beta) \\
&=K^{A_\alpha,B}_{a_n}(y_\beta)\\
&\leq a_n - o(a_n)\\
&= sb_n  - o(a_n)\\
&= s K^{A_\alpha}_{b_n}(y_\beta) - o(a_n).
\end{align*}
Therefore $A_\alpha$ is not optimal, and the claim follows.
\end{proof}

\subsubsection{Generalization to $\R^n$}
In this section, we use Lemma \ref{lem:setRwithoutOptimalOracles} to show that there are sets without optimal Hausdorff oracles in $\R^n$ of every possible dimension. We will need the following lemma on giving sufficient conditions for a product set to have optimal Hausdorff oracles. Interestingly, we need the product formula to hold for arbitrary sets, first proven by Lutz \cite{Lutz17}. Under the assumption that $F$ is regular, the product formula gives
\begin{center}
$\dim_H(F\times G) = \dim_H(F) + \dim_H(G) = \dim_P(F) + \dim_H(G)$,
\end{center}
for every set $G$.

\begin{lem}\label{lem:optimalProduct}
Let $F\subseteq \R^n$ be a set such that $\dim_H(F) = \dim_P(F)$, let $G \subseteq \R^m$ and let $E = F \times G$. Then $E$ has optimal Hausdorff oracles if and only if $G$ has optimal Hausdorff oracles.
\end{lem}	
\begin{proof}
Assume that $G$ has an optimal Hausdorff oracle $A_1$. Let $A_2, A_3$ be Hausdorff oracles for $E$ and $F$, respectively. Let $A \subseteq \N$ be the join of all three oracles. We claim that $A$ is optimal for $E$. Let $B$ be any oracle and let $\epsilon > 0$. Since $A$ is optimal for $G$, by Lemma \ref{lem:joinOptimalisOptimal}, there is a point $z\in G$ such that $\dim^{A,B}(z) \geq \dim_H(G) - \epsilon / 2$ and 
\begin{center}
$K^{A,B}_r(z) \geq K^A_r(z) - \epsilon r / 2$,
\end{center}
for almost every $r$. By the point-to-set principle, we may choose a $y \in F$ such that 
\begin{center}
$\dim^{A, B, z}(y) \geq \dim_H(F) - \epsilon / 2$.
\end{center}
Let $x = (y,z) \in E$. Then
\begin{align*}
K^{A,B}_r(x) &= K^{A,B}_r(y,z)\\
&= K^{A,B}_r(z) + K^{A,B}_r(y \, | \, z)\\
&\geq K^A_r(z) - \epsilon r / 2 + K^{A,B,z}_r(y)\\
&\geq K^A_r(z) - \epsilon r / 2 + (\dim_H(F) - \epsilon / 2)r\\
&= K^A_r(z) - \epsilon r / 2 + (\dim_P(F) - \epsilon / 2)r\\
&\geq K^A_r(z) - \epsilon r / 2 + K^A_r(y) - \epsilon r / 2\\
&\geq K^A_r(z) - \epsilon r / 2 + K^A_r(y \, | \, z) - \epsilon r / 2\\
&\geq K^A_r(y,z) - \epsilon r \\
&= K^A_r(x) - \epsilon r.
\end{align*}
Since $B$ and $\epsilon$ were arbitrary, $A$ is optimal for $E$.

Suppose that $G$ does not have optimal Hausdorff oracles. Let $A$ be a Hausdorff oracle for $E$. It suffices to show that $A$ is not optimal for $E$. Since optimal Hausdorff oracles are closed under the join operation, we may assume that $A$ is a Hausdorff oracle for $F$ and $G$ as well. Since $G$ does not have optimal Hausdorff oracles, there is an oracle $B$ and $\epsilon > 0$ such that, for every $z\in G$, either $\dim^{A,B}(z) < \dim_H(G) - \epsilon$ or
\begin{center}
$K^{A,B}_r(z) < K^A_r(z) - \epsilon r/2$,
\end{center} 
for infinitely many $r$. Let $x \in E$, such that $\dim^{A,B}(x) \geq \dim_H(E) - \epsilon/2$. Let $x = (y,z)$ for some $y\in F$ and $z\in G$. Then we have 
\begin{align*}
\dim_H(F) + \dim_H(G) &= \dim_H(E)\\
&\leq \dim^{A,B}(x) + \epsilon/2\\
&= \dim^{A,B}(y) + \dim^{A,B}(z \, | \, y)+ \epsilon /2\\
&\leq \dim_H(F)  + \dim^{A,B}(z)+ \epsilon/2.
\end{align*}
Hence, $\dim^{A,B}(z) \geq \dim_H(G) - \epsilon/2$. 

We conclude that there are infinitely many $r$ such that
\begin{align*}
K^{A,B}_r(x) &= K^{A,B}_r(z) + K^{A,B}_r(y \, | \, z)\\
&< K^A_r(z) - \epsilon r/2+ K^{A,B}_r(y \, | \, z)\\
&\leq K^A_r(z) - \epsilon r/2+ K^{A}_r(y \, | \, z)\\
&= K^{A,B}_r(x)- \epsilon r/2.
\end{align*}
Thus $E$ does not have optimal Hausdorff oracles.
\end{proof}

\begin{thm}\label{thm:setRNwithoutOptimalOracles}
Assume AC and CH. Then for every $n\in\N$ and $s \in (0, n)$, there is a subset $E\subseteq\R^n$ with $\dim_H(E) = s$ such that $E$ does not have optimal Hausdorff oracles.
\end{thm}
\begin{proof}
We will show this via induction on $n$. For $n = 1$, the conclusion follows from Lemma \ref{lem:setRwithoutOptimalOracles}. 

Suppose the claim holds for all $m < n$. Let $s \in (0, n)$. First assume that $s < n - 1$. Then by our induction hypothesis, there is a set $G \subseteq \R^{n-1}$ without optimal Hausdorff oracles such that $\dim_H(G) = s$. Let $E = \{0\} \times G$. Note that, since $\{0\}$ is a singleton, $\dim_H(\{0\}) =\dim_P(\{0\}) = 0$. Therefore, by Lemma \ref{lem:optimalProduct}, $E$ does not have optimal Hausdorff oracles. By the well-known product formula for Hausdorff dimension,
\begin{align*}
\dim_H(G) &\leq \dim_H(\{0\}) + \dim_H(G)\\ 
&\leq \dim_H(E)\\
&\leq \dim_P(\{0\}) + \dim_H(G)\\
&= \dim_H(G),
\end{align*}
and the claim follows.

We now assume that $s \geq n - 1$. Let $d = s - 1$. By our induction hypothesis, there is a set $G \subseteq \R^{n-1}$ without optimal Hausdorff oracles such that $\dim_H(G) = d$. Let $E = [0,1] \times G$. Note that, since $[0,1]$ has (Lebesgue) measure one, $\dim_H([0,1]) =\dim_P([0,1]) = 1$. Thus, by Lemma \ref{lem:optimalProduct}, $E$ is a set without optimal Hausdorff oracles. By the product formula, 
\begin{align*}
1 + \dim_H(G) &\leq \dim_H([0,1]) + \dim_H(G)\\ 
&\leq \dim_H(E)\\
&\leq \dim_P([0,1]) + \dim_H(G)\\
&= 1 + \dim_H(G),
\end{align*}
and the claim follows.
\end{proof}

\section{Marstrand's Projection Theorem}\label{sec:Marstrand}

The following theorem, due to Lutz and Stull \cite{LutStu18}, gives sufficient conditions for strong lower bounds on the complexity of projected points.
\begin{thm}\label{thm:mainengine}
Let $z \in \R^2$, $\theta \in [0,\pi]$, $C \subseteq \N$, $\eta \in \Q \cap (0, 1) \cap (0, \dim(z))$, $\ve > 0$, and $r \in \N$. Assume the following are satisfied.
\begin{enumerate}
\item For every $s \leq r$, $K_{s}(\theta) \geq s - \log(s)$.
\item $K^{C, \theta}_r(z) \geq K_r(z) - \ve r$.
\end{enumerate}
Then,
\[K^{C, \theta}_r(p_\theta z ) \geq \eta r - \ve r -\frac{4\ve}{1-\eta}r- O(\log r)\,.\]
\end{thm}

The second condition of this theorem requires the oracle $(C, \theta)$ to give essentially no information about $z$. The existence of optimal Hausdorff oracles gives a sufficient condition for this to be true, for all sufficiently large precisions. Thus we are able to show that Marstrands projection theorem holds for any set with optimal Hausdorff oracles.
\begin{thm}
Suppose $E \subseteq \R^2$ has an optimal Hausdorff oracle. Then for almost every $\theta \in [0, \pi]$,
\begin{center}
$\dim_H(p_\theta E) = \min\{\dim_H(E), 1\}$.
\end{center}
\end{thm}
\begin{proof}
Let $A$ be an optimal Hausdorff oracle for $E$. Let $\theta$ be random relative to $A$. Let $B$ be oracle testifying to the point-to-set principle for $p_\theta E$. It suffices to show that
\begin{center}
$\sup\limits_{z\in E} \dim^{A, B}(p_\theta z) = \min\{1, \dim_H(E)\}$.
\end{center}

Since $E$ has optimal Hausdorff oracles, for each $n\in\N$, we may choose a point $z_n \in E$ such that 
\begin{itemize}
\item $\dim^{A,B,\theta}(z_n) \geq \dim_H(E) - \frac{1}{2n}$, and
\item $K^{A,B,\theta}_r(z_n) \geq K^A_r(z_n) - \frac{r}{2n}$ for almost every $r$.
\end{itemize}

Fix a sufficiently large $n$, and let $\ve = 1/2n$. Let $\eta \in \Q$ be a rational such that 
\begin{center}
$\min\{1,\dim_H(E)\} - 5\ve^{1/2} < \eta < 1 - 4\ve^{1/2}$. 
\end{center}

We now show that the conditions of Theorem \ref{thm:mainengine} are satisfied for $\eta, \ve$, relative to $A$. By our choice of $\theta$, 
\begin{center}
$K^A_r(\theta) \geq r - O(\log r)$,
\end{center}
for every $r\in \N$. By our choice of $z_n$ and the Hausdorff optimality of $A$, 
\begin{center}
$K^{A,B, \theta}_r(z_n) \geq K_r(z_n) - \ve r$,
\end{center}
for all sufficiently large $r$. We may therefore apply Theorem \ref{thm:mainengine}, to see that, for all sufficiently large $r$,
\[K^{A, B, \theta}_r(p_\theta z_n ) \geq \eta r - \ve r -\frac{4\ve}{1-\eta}r-O(\log r)\,.\]
Thus,
\begin{align*}
\dim^{A, B}(p_\theta z_n) &\geq \dim^{A, B,\theta}(p_\theta z_n)\\
&= \limsup_r \frac{K^{A, B, \theta}_r(p_\theta z_n )}{r}\\
&\geq \limsup_r \frac{\eta r - \ve r -\frac{4\ve}{1-\eta}r-O(\log r)}{r}\\
&= \limsup_r \eta - \ve -\frac{4\ve}{1-\eta} - o(1)\\
&> \eta - \ve - \ve^{1/2} - o(1)\\
&> \min\{1,\dim_H(E)\} - \ve - 6\ve^{1/2} - o(1).
\end{align*}
Hence,  
\begin{center}
$\lim_{n} \dim^{A, B}(p_\theta z_n) = \min\{1,\dim_H(E)\}$,
\end{center}
and the proof is complete.
\end{proof}

This shows that Marstrand's theorem holds for every set $E$ with $\dim_H(E) = s$ satisfying any of the following:
\begin{enumerate}
\item $E$ is analytic.
\item $\dim_H(E) = \dim_P(E)$.
\item $\mu \ll \mathcal{H}^{s-\delta}$, for every $\delta > 0$ for some metric outer measure $\mu$ such that $0<\mu(E)<\infty$.
\item $\Hs \ll \mu$ and $\Hs(E) > 0$, for some metric outer measure $\mu$ such that $0<\mu(E)<\infty$.
\end{enumerate}
For example, the existence of exact gauged Hausdorff measures on $E$ guarantee the existence of optimal Hausdorff oracles.

\begin{ex}
Let $E$ be a set with $\dim_H(E) = s$ and $\Hs(E) = 0$. Suppose that $0<\mathcal{H}^\phi(E) < \infty$, where $\phi(t) = \frac{t^s}{\log \frac{1}{t}}$. Since $\mathcal{H}^{\phi} \ll  \mathcal{H}^{s-\delta}$ for every $\delta > 0$, Theorem \ref{thm:mainTheorem} implies that $E$ has optimal Hausdorff oracles, and thus Marstrand's theorem holds for $E$.
\end{ex}

\begin{ex}
Let $E$ be a set with $\dim_H(E) = s$ and $\Hs(E) = \infty$. Suppose that $0<\mathcal{H}^\phi(E) < \infty$, where $\phi(t) = t^s \log \frac{1}{t}$. Since $\Hs \ll \mathcal{H}^{\phi}$, Theorem \ref{thm:mainTheorem} implies that $E$ has optimal Hausdorff oracles, and thus Marstrand's theorem holds for $E$.
\end{ex}

\subsection{Counterexample to Marstrand's theorem}\label{ssec:daviesCounter}

In this section we show that there are sets for which Marstrand's theorem does not hold. While not explicitly mentioning optimal Hausdorff oracles, the construction is very similar to the construction in Section \ref{ssec:counterexamOptimal}. 
\begin{thm}
Assuming AC and CH, for every $s\in (0,1)$ there is a set $E$ such that $\dim_H(E) = 1 + s$ but 
\begin{center}
$\dim_H(p_\theta E) = s$
\end{center}
for every $\theta \in (\pi/4, 3\pi/4)$.
\end{thm}
This is a modest generalization of Davies' theorem to sets with Hausdorff dimension strictly greater than one. In the next section we give a new proof of Davies' theorem by generalizing this construction to the endpoint $s = 0$.

We will need the following simple observation.
\begin{obs}
Let $r \in \N$, $s \in (0, 1)$, and $\theta \in (\pi / 8, 3\pi /8)$. Then for every dyadic rectangle
\begin{center}
$R = [d_x - 2^{-r}, d_x + 2^{-r}] \times [d_y - 2^{-sr}, d_y + 2^{-sr}]$,
\end{center}
there is a point $z \in R$ such that $K^\theta_r(p_\theta z) \leq sr + o(r)$.
\end{obs}
\begin{proof}
Note that $p_\theta$ is a Lipschitz function. Thus, for any rectangle 
\begin{center}
$R = [d_x - 2^{-r}, d_x + 2^{-r}] \times [d_y - 2^{-sr}, d_y + 2^{-sr}]$,
\end{center}
The length of its projection (which is an interval) is 
\begin{center}
$\vert p_\theta R \vert \geq c2^{-sr}$
\end{center}
for some constant $c$. It is well known that any interval of length $2^{-\ell}$ contains points $x$ such that
\begin{center}
$K_r(x) \leq \ell r + o(r)$.
\end{center}
\end{proof}
For every $r\in \N$, $\theta \in (\pi / 4, 3\pi / 4)$, binary string $x$ of length $r$ and string $y$ of length $sr$, let $g_\theta(x,y) \mapsto z$ be a function such that 
\begin{center}
$K^\theta_r(p_\theta \, (x,z)) \leq sr + o(r)$.
\end{center}
That is, $g_\theta$, given a rectangle 
\begin{center}
$R = [d_x - 2^{-r}, d_x + 2^{-r}] \times [d_y - 2^{-sr}, d_y + 2^{-sr}]$,
\end{center}
outputs a value $z$ such that $K_r(p_\theta(x, z))$ is small.

Let $s \in (0,1)$. We begin by defining two sequences of natural numbers, $\{a_n\}$ and $\{b_n\}$. Let $a_1 = 2$, and $b_1 = \lfloor 2 / s \rfloor$. Inductively define $a_{n+1} = b_{n}^2$ and $b_{n+1} = \lfloor a_{n+1}/s\rfloor$. We will also need, for every ordinal $\alpha$, a function $f_\alpha : \N \to \{\beta \, | \, \beta < \alpha\}$ such that each ordinal $\beta < \alpha$ is mapped to by infinitely many $n$. Note that such a function exists, since the range is countable assuming CH.

Using AC and CH, we first order the subsets of the natural numbers and we order the angles $\theta \in (\pi/4, 3\pi / 4)$ so that each has at most countably many predecessors.

We will define real numbers $x_\alpha$, $y_\alpha$ and $z_\alpha$ inductively. Let $x_1$ be a real which is random relative to $A_1$. Let $y_1$ be a real which is random relative to $(A_1, x_1)$. Define $z_1$ to be the real whose binary expansion is given by
\begin{align*}
z_1[r] = \begin{cases}
g_{\theta_1}(x_1,y_1)[r] &\text{ if } a_n < r \leq b_n \text{ for some } n \in \N \\
y_1[r] &\text{ otherwise}
\end{cases}
\end{align*}

For the induction step, suppose we have defined our points up to ordinal $\alpha$. Let $x_\alpha$ be a real number which is random relative to the join of $\bigcup_{\beta < \alpha} (A_\beta, x_\beta)$ and $A_\alpha$. Let $y_\alpha$ be random relative to the join of $\bigcup_{\beta < \alpha} (A_\beta, x_\beta)$, $A_\alpha$ and $x_\alpha$. This is possible, as we are assuming CH, and so this union is countable. Let $z_\alpha$ be the point whose binary expansion is given by
\begin{align*}
z_\alpha[r] = \begin{cases}
g_{\theta_\beta}(x_\alpha, y_\alpha)[r] &\text{ if } a_n < r \leq b_n, \text{ for } f_\alpha(n) = \beta \\
y_\alpha[r] &\text{ otherwise}
\end{cases}
\end{align*}

Finally, we define our set $E = \{(x_\alpha, z_\alpha)\}$. 

\begin{lem}
For every $\theta \in (\pi/4, 3\pi/4)$,
\begin{center}
$\dim_H(p_\theta E) \leq s$
\end{center}
\end{lem}
\begin{proof}
Let $\theta \in (\pi / 4, 3\pi /4)$ and $\alpha$ be its corresponding ordinal. Let $A$ be an oracle encoding $\theta$ and 
\begin{center}
$\bigcup\limits_{\beta \leq \alpha} (x_\beta, y_\beta, z_\beta)$.
\end{center}
Note that, since we assumed CH, this is a countable union, and so the oracle is well defined.

Let $z = (x_\beta, z_\beta) \in E$. First assume that $\beta \leq \alpha$. Then, by our construction of $A$, all the information of $p_\theta z$ is already encoded in our oracle, and so
\begin{center}
$K^A_r(p_\theta z) = o(r)$.
\end{center}

Now assume that $\beta > \alpha$. Then by our construction of $E$, there are infinitely many $n$ such that $f_\beta(n) = \alpha$. Therefore there are infinitely many $n$ such that
\begin{center}
$z_\beta [r] = g_{\theta_\alpha}(x_\beta, y_\beta)[r]$,
\end{center}
for $a_n < r \leq b_n$. Recalling the definition of $g_{\theta_\alpha}$, this means that, for each such $n$,
\begin{center}
$K^\theta_{b_n}(p_\theta z) = sb_n + o(r)$.
\end{center}

Therefore, by the point-to-set principle,
\begin{align*}
\dim_H(p_\theta E) &\leq \sup_{z\in E} \dim^A(p_\theta z)\\
&\leq \sup_{\beta > \alpha} \liminf_n \frac{K^A_{b_n}(p_\theta z)}{b_n}\\
&\leq \sup_{\beta > \alpha} \liminf_n \frac{sb_n}{b_n}\\
&= s,
\end{align*}
and the proof is complete.
\end{proof}

\begin{lem}
The Hausdorff dimension of $E$ is $1+s$.
\end{lem}
\begin{proof}
We first give an upper bound on the dimension. Let $A$ be an oracle encoding $\theta_1$. Let $z = (x_\alpha, z_\alpha)$. By our construction of $E$, there are infinitely many $n$ such that $f_\alpha(n) = 1$. Therefore there are infinitely many $n$ such that
\begin{center}
$z_\beta [r] = g_{\theta_1}(x_\beta, y_\beta)[r]$,
\end{center}
for $a_n < r \leq b_n$. Recalling the definition of $g_{\theta_1}$, this means that, for each such $n$,
\begin{center}
$K^{\theta_1}_{b_n}(p_{\theta_1} z) = sb_n + o(r)$.
\end{center}
Moreover,
\begin{align*}
K^{\theta_1}_{b_n}(x_\alpha, z_\alpha) &\leq K^{\theta_1}_{b_n}(x_\alpha) + K^{\theta_1}_{b_n}(z_\alpha \mid x_\alpha)+ o(r)\\
&\leq b_n + K^{\theta_1}_{b_n}(p_{\theta_1} z)+ o(r)\\
&\leq b_n + sb_n + o(b_n)).
\end{align*}

Therefore, by the point-to-set principle,
\begin{align*}
\dim_H(E) &\leq \sup_{z\in E} \dim^A(z)\\
&\leq \sup_{z\in E} \liminf_n \frac{K^A_{b_n}(z)}{b_n}\\
&\leq \sup_{z\in E} \liminf_n \frac{(1 + s) b_n + o(b_n)}{b_n}\\
&= 1 + s.
\end{align*}

For the upper bound, let $A$ be a Hausdorff oracle for $E$, and let $\alpha$ be the ordinal corresponding to $A$. By construction of $z = (x_\alpha, z_\alpha)$, 
\begin{center}
$K^A_r(x_\alpha) \geq r - o(r)$,
\end{center}
for all $r\in \N$. We also have, for every $n$,
\begin{align*}
K^A_{a_n}(z_\alpha \, | \, x_\alpha) &\geq K^A_{a_n}(y_\alpha \, | \, x_\alpha) - b_{n-1} - o(a_n)\\
&\geq a_n - b_{n-1} - o(a_n)\\
&= a_n - a^{\frac{1}{2}}_n - o(a_n).
\end{align*}
Hence, for every $n$ and every $a_n < r \leq b_n$,
\begin{align*}
K^A_{r}(z_\alpha \, | \, x_\alpha) &\geq K^A_{a_n}(z_\alpha \, | \, x_\alpha) \\
&\geq a_n - a^{\frac{1}{2}}_n - o(a_n).
\end{align*}
This implies that
\begin{align*}
\frac{K^A_r(x_\alpha, z_\alpha)}{r} &= \frac{K^A_r(x_\alpha) +K^A_{r}(z_\alpha \, | \, x_\alpha)}{r}\\
&\geq \frac{r +a_n - a^{\frac{1}{2}}_n - o(a_n)}{r}\\
&= 1 + s- o(1).
\end{align*}

We can also conclude that, for every $n$ and every $b_n < r \leq a_{n+1}$,
\begin{align*}
K^A_{r}(z_\alpha \, | \, x_\alpha) &\geq K^A_{b_n}(z_\alpha \, | \, x_\alpha) K^A_{a_n,b_n}(z_\alpha \, | \, x_\alpha)-o(r)\\
&\geq a_n - a^{\frac{1}{2}}_n + r - b_n - o(r).
\end{align*}
This implies that
\begin{align*}
\frac{K^A_r(x_\alpha, z_\alpha)}{r} &= \frac{K^A_r(x_\alpha) +K^A_{r}(z_\alpha \, | \, x_\alpha)}{r}\\
&\geq \frac{r +a_n - a^{\frac{1}{2}}_n + r - b_n - o(r)}{r}\\
&\geq 1+s -o(1).
\end{align*}

These inequalities, combined with the point-to-set principle show that
\begin{align*}
\dim_H(E) &= \sup_{z\in E} \dim^A(z)\\
&\geq \sup_{z\in E} \liminf_r \frac{K^A_{r}(z)}{r}\\
&\geq \sup_{z\in E} \liminf_r 1+s\\
&= 1 + s,
\end{align*}
and the proof is complete.
\end{proof}

\subsection{Generalization to the endpoint $s = 0$}
\begin{thm}
Assuming AC and CH, there is a set $E$ such that $\dim_H(E) = 1$ but 
\begin{center}
$\dim_H(p_\theta E) = 0$
\end{center}
for every $\theta \in (\pi/4, 3\pi/4)$.
\end{thm}

For every $r\in \N$, $\theta \in (\pi / 4, 3\pi / 4)$, binary string $x$ of length $r$ and string $y$ of length $sr$, let $g^s_\theta(x,y) \mapsto z$ be a function such that 
\begin{center}
$K^\theta_r(p_\theta \, (x,z)) \leq sr + o(r)$.
\end{center}
That is, $g^s_\theta$, given a rectangle 
\begin{center}
$R = [d_x - 2^{-r}, d_x + 2^{-r}] \times [d_y - 2^{-sr}, d_y + 2^{-sr}]$,
\end{center}
outputs a value $z$ such that $K_r(p_\theta(x, z))$ is small.

We begin by defining two sequences of natural numbers, $\{a_n\}$ and $\{b_n\}$. Let $a_1 = 2$, and $b_1 = 4$. Inductively define $a_{n+1} = b_{n}^2$ and $b_{n+1} = (n+1)\lfloor a_{n+1}\rfloor$. We will also need, for every ordinal $\alpha$, a function $f_\alpha : \N \to \{\beta \, | \, \beta < \alpha\}$ such that each ordinal $\beta < \alpha$ is mapped to by infinitely many $n$. Note that such a function exists, since the range is countable assuming CH.

Using AC and CH, we first order the subsets of the natural numbers and we order the angles $\theta \in (\pi/4, 3\pi / 4)$ so that each has at most countably many predecessors.

We will define real numbers $x_\alpha$, $y_\alpha$ and $z_\alpha$ inductively. Let $x_1$ be a real which is random relative to $A_1$. Let $y_1$ be a real which is random relative to $(A_1, x_1)$. Define $z_1$ to be the real whose binary expansion is given by
\begin{align*}
z_1[r] = \begin{cases}
g^1_{\theta_1}(x_1,y_1)[r] &\text{ if } a_n < r \leq b_n \text{ for some } n \in \N \\
y_1[r] &\text{ otherwise}
\end{cases}
\end{align*}

For the induction step, suppose we have defined our points up to ordinal $\alpha$. Let $x_\alpha$ be a real number which is random relative to the join of $\bigcup_{\beta < \alpha} (A_\beta, x_\beta)$ and $A_\alpha$. Let $y_\alpha$ be random relative to the join of $\bigcup_{\beta < \alpha} (A_\beta, x_\beta)$, $A_\alpha$ and $x_\alpha$. This is possible, as we are assuming CH, and so this union is countable. Let $z_\alpha$ be the point whose binary expansion is given by
\begin{align*}
z_\alpha[r] = \begin{cases}
g^{1/n}_{\theta_\beta}(x_\alpha, y_\alpha)[r] &\text{ if } a_n < r \leq b_n, \text{ for } f_\alpha(n) = \beta \\
y_\alpha[r] &\text{ otherwise}
\end{cases}
\end{align*}

Finally, we define our set $E = \{(x_\alpha, z_\alpha)\}$. 

\begin{lem}
For every $\theta \in (\pi/4, 3\pi/4)$,
\begin{center}
$\dim_H(p_\theta E) = 0$.
\end{center}
\end{lem}
\begin{proof}
Let $\theta \in (\pi / 4, 3\pi /4)$ and $\alpha$ be its corresponding ordinal. Let $A$ be an oracle encoding $\theta$ and 
\begin{center}
$\bigcup\limits_{\beta \leq \alpha} (x_\beta, y_\beta, z_\beta)$.
\end{center}
Note that, since we assumed CH, this is a countable union, and so the oracle is well defined.

Let $z = (x_\beta, z_\beta) \in E$. First assume that $\beta \leq \alpha$. Then, by our construction of $A$, all the information of $p_\theta z$ is already encoded in our oracle, and so
\begin{center}
$K^A_r(p_\theta z) = o(r)$.
\end{center}

Now assume that $\beta > \alpha$. Then by our construction of $E$, there are infinitely many $n$ such that $f_\beta(n) = \alpha$. Therefore there are infinitely many $n$ such that
\begin{center}
$z_\beta [r] = g^{1/n}_{\theta_\alpha}(x_\beta, y_\beta)[r]$,
\end{center}
for $a_n < r \leq b_n$. Recalling the definition of $g^{1/n}_{\theta_\alpha}$, this means that, for each such $n$,
\begin{center}
$K^\theta_{b_n}(p_\theta z) = \frac{b_n}{n} + o(r)$.
\end{center}

Therefore, by the point-to-set principle,
\begin{align*}
\dim_H(p_\theta E) &\leq \sup_{z\in E} \dim^A(p_\theta z)\\
&\leq \sup_{\beta > \alpha} \liminf_n \frac{K^A_{b_n}(p_\theta z)}{b_n}\\
&\leq \sup_{\beta > \alpha} \liminf_n \frac{\frac{b_n}{n}}{b_n}\\
&= \frac{1}{n},
\end{align*}
and the proof is complete.
\end{proof}

\begin{lem}
The Hausdorff dimension of $E$ is $1$.
\end{lem}
\begin{proof}
We first give an upper bound on the dimension. Let $A$ be an oracle encoding $\theta_1$. Let $z = (x_\alpha, z_\alpha)$. By our construction of $E$, there are infinitely many $n$ such that $f_\alpha(n) = 1$. Therefore there are infinitely many $n$ such that
\begin{center}
$z_\beta [r] = g^{1/n}_{\theta_1}(x_\beta, y_\beta)[r]$,
\end{center}
for $a_n < r \leq b_n$. Recalling the definition of $g^{1/n}_{\theta_1}$, this means that, for each such $n$,
\begin{center}
$K^{\theta_1}_{b_n}(p_{\theta_1} z) = \frac{b_n}{n} + o(r)$.
\end{center}
Moreover,
\begin{align*}
K^{\theta_1}_{b_n}(x_\alpha, z_\alpha) &\leq K^{\theta_1}_{b_n}(x_\alpha) + K^{\theta_1}_{b_n}(z_\alpha \mid x_\alpha)+ o(r)\\
&\leq b_n + K^{\theta_1}_{b_n}(p_{\theta_1} z)+ o(r)\\
&\leq b_n + \frac{b_n}{n} + o(b_n)).
\end{align*}

Therefore, by the point-to-set principle,
\begin{align*}
\dim_H(E) &\leq \sup_{z\in E} \dim^A(z)\\
&\leq \sup_{z\in E} \liminf_n \frac{K^A_{b_n}(z)}{b_n}\\
&\leq \sup_{z\in E} \liminf_n \frac{(b_n + b_n/n + o(b_n)}{b_n}\\
&= 1.
\end{align*}

For the upper bound, let $A$ be a Hausdorff oracle for $E$, and let $\alpha$ be the ordinal corresponding to $A$. By construction of $z = (x_\alpha, z_\alpha)$, 
\begin{center}
$K^A_r(x_\alpha) \geq r - o(r)$,
\end{center}
for all $r\in \N$. We also have, for every $n$,
\begin{align*}
K^A_{a_n}(z_\alpha \, | \, x_\alpha) &\geq K^A_{a_n}(y_\alpha \, | \, x_\alpha) - b_{n-1} - o(a_n)\\
&\geq a_n - b_{n-1} - o(a_n)\\
&= a_n - a^{\frac{1}{2}}_n - o(a_n).
\end{align*}
Hence, for every $n$ and every $a_n < r \leq b_n$,
\begin{align*}
K^A_{r}(z_\alpha \, | \, x_\alpha) &\geq K^A_{a_n}(z_\alpha \, | \, x_\alpha) \\
&\geq a_n - a^{\frac{1}{2}}_n - o(a_n).
\end{align*}
This implies that
\begin{align*}
\frac{K^A_r(x_\alpha, z_\alpha)}{r} &= \frac{K^A_r(x_\alpha) +K^A_{r}(z_\alpha \, | \, x_\alpha)}{r}\\
&\geq \frac{r +a_n - a^{\frac{1}{2}}_n - o(a_n)}{r}\\
&= 1- o(1).
\end{align*}

We can also conclude that, for every $n$ and every $b_n < r \leq a_{n+1}$,
\begin{align*}
K^A_{r}(z_\alpha \, | \, x_\alpha) &\geq K^A_{b_n}(z_\alpha \, | \, x_\alpha) + K^A_{a_n,b_n}(z_\alpha \, | \, x_\alpha)-o(r)\\
&\geq a_n - a^{\frac{1}{2}}_n + r - b_n - o(r).
\end{align*}
This implies that
\begin{align*}
\frac{K^A_r(x_\alpha, z_\alpha)}{r} &= \frac{K^A_r(x_\alpha) +K^A_{r}(z_\alpha \, | \, x_\alpha)}{r}\\
&\geq \frac{r +a_n - a^{\frac{1}{2}}_n + r - b_n - o(r)}{r}\\
&\geq 1 -o(1).
\end{align*}

These inequalities, combined with the point-to-set principle show that
\begin{align*}
\dim_H(E) &= \sup_{z\in E} \dim^A(z)\\
&\geq \sup_{z\in E} \liminf_r \frac{K^A_{r}(z)}{r}\\
&\geq \sup_{z\in E} 1\\
&= 1,
\end{align*}
and the proof is complete.
\end{proof}

\section{Optimal Packing Oracles}\label{sec:optimalPacking}
Similarly, we can define optimal \textit{packing} oracles for a set.
\begin{defn}
Let $E \subseteq\R^n$ and $A \subseteq \N$. We say that $A$ is an \textit{optimal packing oracle} (or \textit{packing optimal}) for $E$ if the following conditions are satisfied.
\begin{enumerate}
\setlength\itemsep{.5em}
\item $A$ is a packing oracle for $E$.
\item For every $B \subseteq \N$ and every $\epsilon > 0$ there is a point $x\in E$ such that $\Dim^{A,B}(x) \geq \dim_P(E) - \epsilon$ and for almost every $r\in \N$
\begin{center}
$K^{A,B}_r(x) \geq K^A_r(x) - \epsilon r$.
\end{center}
\end{enumerate}
\end{defn}

Let $E\subseteq \R^n$ and $A\subseteq\N$. For $B\subseteq \N$, $\epsilon > 0$ define the set
\begin{center}
$N(A, B,\epsilon) = \{x \in E \, | \, (\forall^\infty r) \, K^{A,B}_r(x) \geq K^A_r(x) - \epsilon r\}$
\end{center}

\begin{prop}
Let $E\subseteq\R^n$ be a set such that $\dim_P(E) > 0$ and let $A$ be an oracle. Then $A$ is packing optimal for $E$ if and only if $A$ is a packing oracle and for every $B\subseteq\N$ and $\epsilon > 0$, $\dim_P(N(A, B,\epsilon)) = \dim_P(E)$.
\end{prop}
\begin{proof}
For the forward direction, let $A$ be an optimal packing oracle for $E$. Then by the first condition of the definition, $A$ is a packing oracle. Let $B\subseteq\N$ and $\epsilon > 0$. Let $C$ be a packing oracle for $N(A, B,\epsilon)$. For the sake of contradiction, suppose that 
\begin{center}
$\dim_P(N(A, B,\epsilon)) < \dim_P(E) - \gamma$,
\end{center}
for some $\gamma > 0$. We will, without loss of generality, assume that $\gamma < \epsilon$. Then, by the point to set principle, for every $x \in N(A,B,\epsilon)$,
\begin{align*}
\Dim^{A, (B, C)}(x) &\leq \Dim^{C}(x)\\
&\leq \dim_P(N(A,B,\epsilon))\\
&< \dim_P(E) - \gamma.
\end{align*}
Since, $A$ is an optimal packing oracle for $E$, there is a point $x \in E$ such that $\Dim^{A, (B,C)}(x) \geq \dim_P(E) - \gamma$ and for almost every $r\in \N$
\begin{center}
$K^{A,(B, C)}_r(x) \geq K^A_r(x) - \gamma r$.
\end{center}
By our previous discussion, any such point $x$ cannot be in $N(A,B,\epsilon)$. However, if $x \notin N(A,B,\epsilon)$, then for infinitely many $r$,
\begin{center}
$K^{A,(B, C)}_r(x) < K^A_r(x) - \epsilon r$.
\end{center}
Thus, no such $x$ exists, contradicting the fact that $A$ is packing optimal.

For the backward direction, let $A$ be an oracle satisfying the hypothesis. Then $A$ is a Hausdorff oracle for $E$ and the first condition of optimal Hausdorff oracles is satisfied. Let $B \subseteq \N$ and $\epsilon > 0$. By our hypothesis and the point-to-set principle,
\begin{align*}
\dim_H(E) &= \dim_H(N(A, B,\epsilon))\\
&\leq \sup\limits_{x \in N(A, B,\epsilon)} \dim^{A,B}(x).
\end{align*}
Therefore, there is certainly a point $x \in E$ such that $\dim^{A,B}(x) \geq \dim_H(E) - \epsilon$ and 
\begin{center}
$K^{A,B}_r(x) \geq K^A_r(x) - \epsilon r$,
\end{center}
for almost every $r\in\N$.
\end{proof}

\begin{lem}\label{lem:joinPackingOptimalisOptimal}
Let $E\subseteq \R^n$. If $A$ is packing optimal for $E$, then the join $C = (A,B)$ is packing optimal for $E$ for every oracle $B$.
\end{lem}
\begin{proof}
Let $A$ be an optimal packing oracle for $E$, let $B$ be an oracle and let $C =(A,B)$. By the point-to-set principle (Theorem \ref{thm:p2s}),
\begin{align*}
\dim_P(E) &= \sup\limits_{x\in E} \Dim^A(x)\\
&\geq \sup\limits_{x\in E} \Dim^{C}(x)\\
&\geq \dim_P(E).
\end{align*}
Hence, the oracle $C = (A,B)$ is a packing oracle for $E$.

Let $B^\prime \subseteq \N$ be an oracle, and let $\epsilon > 0$. Let $x \in E$ be a point such that
\begin{equation}
\Dim^{A, (B, B^\prime)}(x) \geq \dim_P(E) - \epsilon / 2,
\end{equation}
and
\begin{equation}
K_r^{A, (B, B^\prime)}(x) \geq K^A_r(x) - \epsilon r / 2,
\end{equation}
for almost every precision $r$. Note that such a point exists since $A$ is packing optimal for $E$.

For all sufficiently large $r$,
\begin{align*}
K^{(A, B), B^\prime}_r(x) &= K^{A, (B, B^\prime)}_r(x)\\
&\geq K^{A}_r(x) - \epsilon r/2\\
&\geq K^{A, B}_r(x) - \epsilon r/2\\
&= K^{C}_r(x) - \epsilon r/2.
\end{align*}
Therefore, $C = (A,B)$ is packing optimal for $E$.
\end{proof}

We now give some basic closure properties of the class of sets with optimal packing oracles. 
\begin{obs}\label{obs:optimalPackingSubset}
Let $F \subseteq E$. If $\dim_P(F) = \dim_P(E)$ and $F$ has an optimal packing oracle, then $E$ has an optimal packing oracle.
\end{obs}

We can also show that having optimal packing oracles is closed under countable unions.
\begin{lem}\label{lem:optimalPackingOraclesClosedUnderUnion}
Let $E_1,E_2,\ldots$ be a countable sequence of sets and let $E = \cup_n E_n$. If every set $E_n$ has an optimal packing oracle, then $E$ has an optimal Hausdorff oracle.
\end{lem}
\begin{proof}
We first note that 
\begin{center}
$\dim_P(E) = \sup_n \dim_P(E_n)$.
\end{center}
For each $n$, let $A_n$ be an optimal packing oracle for $E_n$. Let $A$ be the join of $A_1, A_2,\ldots$. Let $B$ be an oracle guaranteed by Theorem \ref{thm:p2s} such that
\begin{center}
$\sup_x \Dim^B(x) = \sup_n \dim_P(E_n)$.
\end{center} 
Note that, by Lemma \ref{lem:joinOptimalisOptimal}, for every $n$, $(A,B)$ is packing optimal for $E_n$.

We now claim that $(A, B)$ is an optimal packing oracle for $E$. Theorem \ref{thm:p2s} shows that item (1) of the definition of optimal packing oracles is satisfied. For item (2), let $C \subseteq \N$ be an oracle, and let $\epsilon > 0$. Let $n$ be a number such that $\dim_P(E_n) > \dim_P(E) - \epsilon$. Since $(A,B)$ is packing optimal for $E_N$, there is a point $x \in E_n$ such that
\begin{enumerate}
\item[(i)] $\dim^{(A, B), C}(x) \geq \dim_P(E_n) - \epsilon \geq \dim_P(E) - \epsilon$, and
\item[(ii)] for almost every $r$,
\begin{center}
$K^{(A, B), C}_r(x) \geq K^{(A,B)}_r(x) - \epsilon r$.
\end{center}
\end{enumerate}
Therefore, item (2) of the definition of optimal packing oracles is satisfied, and so $(A,B)$ is Hausdorff optimal for $E$.
\end{proof}

For every $0 \leq \alpha < \beta \leq 1$ define the set 
\begin{center}
$D_{\alpha, \beta} = \{ x \in (0,1) \, | \, \dim(x) = \alpha \text{ and } \Dim(x) = \beta\}$.
\end{center}
\begin{lem}\label{lem:OscillatingHasOptimal}
For every $0\leq  \alpha < \beta \leq 1$, $D_{\alpha, \beta}$ has optimal Hausdorff and optimal packing oracles and 
\begin{align*}
&\dim_H(D_{\alpha, \beta}) = \alpha\\
&\dim_P(D_{\alpha, \beta}) = \beta.
\end{align*}
\end{lem}
\begin{proof}
We begin by noting that $D_{\alpha, \beta}$ is Borel. Therefore, by Theorems \ref{thm:mainTheorem} and \ref{thm:mainTheoremPacking}, $D_{\alpha, \beta}$ has optimal Hausdorff and optimal packing oracles. Thus, it suffices to show prove the dimension equalities.

Define the increasing sequence of natural numbers $\{h_j\}$ inductively as follows. Let $h_1 = 2$, and let $h_{j+1} = 2^{h_j}$. For every oracle $A$ let $z_A$ be a point such that, for every $\delta > 0$ and all sufficiently large $r$,
\begin{center}
$K^{A}_{(1 + \delta) r, r}(z_A \, | \, z_A) = \alpha \delta r = K_{(1 + \delta) r, r}(z_A \, | \, z_A)$.
\end{center}
Let $y_A$ be random relative to $A$ and $z_A$.

Let $x_A$ be the point whose binary expansion is given by
\begin{align*}
x_A[r] = \begin{cases}
z_A[r] &\text{ if } h_j < r \leq \frac{1-\beta}{1-\alpha}h_{j+1} \text{ for some } j \in \N \\
y_A[r] &\text{ otherwise}
\end{cases}
\end{align*}

Let $A$ be an oracle, and consider the point $x_A$. Let $r \in \N$ be sufficiently large and let $j\in\N$ such that $h_j < r \leq h_{j+1}$. We first suppose that $r \leq \frac{1-\beta}{1-\alpha}h_{j+1}$. Then
\begin{align*}
K_{r}(x_A) &\geq K^A_{r}(x_A) &= K^A_{h_{j}}(x_A) + K^A_{r,h_{j}}(x_A\,|\, x_A)\\
&= O(\log r)+ K^A_{r,h_{j-1}}(z_A\,|\, z_A)\\
&= \alpha r + O(\log r)\\
&\geq K_{r}(x_A).
\end{align*}

Now suppose that $r > \frac{1-\beta}{1-\alpha}h_{j+1}$. Let $t = \frac{1-\beta}{1-\alpha}h_{j+1}$. Then
\begin{align*}
K_r(x_A) &\geq K^A_{r}(x_A)\\
&= K^A_{t}(x_A) + K^A_{r,t}(x_A\,|\, x_A)+ O(\log r)\\
&= \alpha t + K^A_{r,t}(x_A\,|\, x_A)+ O(\log r)\\
&= \alpha t + r - t + O(\log r)\\
&= r - t(1 - \alpha) + O(\log r)\\
&= r - (1-\beta) h_{j+1} +O(\log r)\\
&\geq K_r(x_A).
\end{align*}
In particular, $K^A_{r}(x_A)\geq \alpha r$ for every $h_j < r \leq h_{j+1}$. Hence for every oracle $A$,
\begin{center}
$\dim^A(x_A) = \alpha = \dim(x_A)$.
\end{center}

For all sufficiently large $j$, 
\begin{align*}
K_{h_j}(x_A) &= K^A_{r}(x_A)\\\
&= h_j - (1-\beta) h_{j} +O(\log r)\\
&= \beta h_{j} +O(\log r),
\end{align*}
and so 
\begin{center}
$\Dim^A(x_A) = \beta = \Dim(x_A)$.
\end{center}
Therefore, for every $A$, $x_A \in D_{\alpha,\beta}$. 

Finally, by the above bounds,
\begin{align*}
&\dim_H(D_{\alpha, \beta}) = \alpha\\
&\dim_P(D_{\alpha, \beta}) = \beta.
\end{align*}
\end{proof}

\subsection{Sufficient conditions for optimal packing oracles}
\begin{lem}
Let $E\subseteq\R^n$ be a set such that $\dim_H(E) = \dim_P(E) = s$. Then $E$ has optimal Hausdorff and optimal packing oracles.
\end{lem}
\begin{proof}
Lemma \ref{lem:regularHasOptimalHaus} shows that $E$ has optimal Hausdorff oracles. Let $A_1$ be an optimal Hausdorff oracle for $E$. Let $A_2$ be a packing oracle for $E$. Let $A = (A_1, A_2)$. By Lemma \ref{lem:joinOptimalisOptimal}, $A$ is an optimal Hausdorff oracle for $E$. We now show that $A$ is an optimal packing oracle for $E$.

It is clear that $A$ is a packing oracle for $E$. Let $B\subseteq\N$ and $\epsilon > 0$. Since $A$ is Hausdorff optimal for $E$, there is a point $x \in E$ such that $\dim^{A,B}(x) \geq s - \epsilon$ and 
\begin{center}
$K^{A,B}_r(x) \geq K^A_r(x) - \epsilon r$
\end{center}
for almost every $r$. Therefore 
\begin{align*}
\Dim^{A,B}(x) &\geq \dim^{A,B}(x)\\
&\geq s - \epsilon\\
&=\dim_P(E) - \epsilon.
\end{align*}
Therefore $x$ satisfies the second condition of optimal packing oracles, and the conclusion follows.
\end{proof}

\begin{thm}\label{thm:mainTheoremPacking}
Let $E\subseteq\R^n$ with $\dim_P(E) = s$. Suppose there is a metric outer measure $\mu$ such that
\begin{center}
$0<\mu(E)<\infty$,
\end{center}  
and either 
\begin{enumerate}
\item $\mu \ll \Ps$, or
\item $\Ps \ll \mu$ and $\Ps(E) > 0$.
\end{enumerate}
Then $E$ has an optimal packing oracle $A$.
\end{thm}
\begin{proof}
Let $A\subseteq\N$ be a packing oracle for $E$ such that $p_{\mu,E}$ is computable relative to $A$. Note that such an oracle exists by the point-to-set principle and routine encoding. We will show that $A$ is packing optimal for $E$. 

For the sake of contradiction, suppose that there is an oracle $B$ and $\epsilon > 0$ such that, for every $x \in E$ either
\begin{enumerate}
\item $\Dim^{A,B}(x) < s - \epsilon$, or
\item there are infinitely many $r$ such that $K^{A,B}_r(x) < K^A_r(x) - \epsilon r$.
\end{enumerate}

Let $N$ be the set of all $x$ for which the second item holds. By Lemma \ref{lem:mainTheoremEngine}, $\mu(N) = 0$. We also note that, by the point-to-set principle, 
\begin{center}
$\Dim_H(E - N) \leq s - \epsilon$,
\end{center}
and so $\Ps(E - N) = 0$.

To achieve the desired contradiction, we first assume that $\mu \ll \Ps$. In this case, it suffices to show that $\mu(E-N) > 0$.  Since $\mu \ll \Ps$,
\begin{center}
$\mu(E - N) = 0$.
\end{center}
Since $\mu$ is a metric outer measure, 
\begin{align*}
0 &< \mu(E)\\
&\leq \mu(N) + \mu(E-N)\\
&= 0,
\end{align*}
a contradiction. 

Now suppose that $\Ps \ll \mu$ and $\Ps(E) > 0$. Then, since $\Ps$ is an outer measure, $\Ps(E) > 0$ and $\Ps(E - N) = 0$ we must have $\Ps(N) >0$. However this implies that  $\mu(N) > 0$, and we again have the desired contradiction. Thus $A$ is an optimal packing oracle for $E$ and the proof is complete.
\end{proof}

We now show that every analytic set has optimal packing oracles. 
\begin{lem}\label{lem:analyticHasOptimalPacking}
Every analytic set $E$ has optimal packing oracles.
\end{lem}
\begin{proof}
A set $E \subseteq \R^n$ is called an packing $s$-set if
\begin{center}
$0 < \Ps(E) <\infty$.
\end{center}
Since $\Ps$ is a metric outer measure, and trivially absolutely continuous with respect to itself, Theorem \ref{thm:mainTheoremPacking} shows that if $E$ is a packing $s$-set then there is an optimal packing oracle for $E$.

Now assume that $E$ is compact, and let $s = \dim_H(E)$. Then for every $t < s$, $\mathcal{H}^t(E) > 0$. Thus, by Theorem \ref{thm:compactSset}, there is a sequence of compact subsets $F_1,F_2,\ldots$ of $E$ such that 
\begin{center}
$\dim_P(\bigcup_n F_n) = \dim_P(E)$,
\end{center}
and, for each $n$,
\begin{center}
$0 < \mathcal{P}^{s_n}(F_n) < \infty$,
\end{center}
where $s_n = s - 1/n$. Therefore, each set $F_n$ has optimal packing oracles. Hence, by Lemma \ref{lem:optimalPackingOraclesClosedUnderUnion}, $E$ has optimal packing oracles and the conclusion follows.

We now show that every $\Sigma^0_2$ set has optimal packing oracles. Suppose $E = \cup_n F_n$ is $\Sigma^0_1$, where each $F_n$ is compact. As we have just seen, each $F_n$ has optimal packing oracles. Therefore, by Lemma \ref{lem:optimalPackingOraclesClosedUnderUnion}, $E$ has optimal packing oracles and the conclusion follows. 

Finally, let $E$ be analytic. By Theorem \ref{thm:compactSset}, there is a $\Sigma^0_2$ subset $F$ of the same packing dimension as $E$. We have just seen that $F$ must have an optimal packing oracle. Since $\dim_P(F) = \dim_P(E)$, by Observation \ref{obs:optimalPackingSubset} $E$ has optimal packing oracles, and the proof is complete
\end{proof}

\subsection{Sets without optimal oracles}

\begin{thm}\label{thm:noHausnoPackingOptimal}
Assuming CH and AC, for every $0 < s_1 < s_2 \leq 1$ there is a set $E \subseteq \R$ which does not have Hausdorff optimal nor packing optimal oracles such that
\begin{center}
$\dim_H(E) = s_1$ and $\dim_P(E) = s_2$.
\end{center}
\end{thm}
\begin{proof}
Let $\delta = s_2 - s_1$. We begin by defining two sequences of natural numbers, $\{a_n\}$ and $\{b_n\}$. Let $a_1 = 2$, and $b_1 = 4$. Inductively define $a_{n+1} = 2^{b_n}$ and $b_{n+1} = 2^{a_{n+1}}$.

Using AC and CH, we order the subsets of the natural numbers such that every subset has countably many predecessors. For every countable ordinal $\alpha$, let $f_\alpha : \N \to \{\beta \, | \, \beta < \alpha\}$ be a function such that each ordinal $\beta$ strictly less than $\alpha$ is mapped to by infinitely many $n$. Note that such a function exists, since the range is countable assuming CH.

We will define real numbers $w_\alpha$, $x_\alpha$, $y_\alpha$ and $z_\alpha$ via transfinite induction. Let $x_1$ be a real such that, for every $\gamma > 0$ and all sufficiently large $r$,
\begin{center}
$K^{A_1}_{(1 + \gamma) r, r}(w_1 \, | \, w_1) = s_1 \gamma r = K_{(1 + \gamma) r, r}(w_1 \, | \, w_1)$.
\end{center}
Let $x_1$ be random relative to $A_1$ and $w_1$. Let $y_1$ be a real such that, for every $\gamma > 0$ and all sufficiently large $r$,
\begin{center}
$K^{A_1}_{(1 + \gamma) r, r}(y_1 \, | \, y_1) = (s_1 + \frac{\delta}{2}) \gamma r = K_{(1 + \gamma) r, r}(y_1 \, | \, y_1)$.
\end{center}
Let $z_1$ be the real whose binary expansion is given by
\begin{align*}
z_1[r] = \begin{cases}
w_1[r] &\text{ if } a_n < r \leq \frac{1-s_2}{1-s_1}b_n \text{ for some }  n \in \N \\
x_1[r] &\text{ if } \frac{1-s_2}{1-s_1}b_n < r \leq b_n  \text{ for some }  n \in \N \\
y_1[r] &\text{ if } b_n < r \leq (1-\delta)a_{n+1} < \text{ for some }  n \in \N \\
0 &\text{ if } (1-\delta)a_{n+1} < r \leq (1+\delta)a_{n+1} < \text{ for some } n \in \N \\
\end{cases}
\end{align*}

For the induction step, suppose we have defined our points up to $\alpha$. Let $A$ be the join of $ \bigcup_{\beta < \alpha} (A_\beta, w_\beta. x_\beta, y_\beta, z_\beta)$ and $A_\alpha$. Let $x_\alpha$ be a real such that, for every $\gamma > 0$ and all sufficiently large $r$,
\begin{center}
$K^{A}_{(1 + \gamma) r, r}(w_\alpha \, | \, w_\alpha) = s_1 \gamma r = K_{(1 + \gamma) r, r}(w_\alpha \, | \, w_\alpha)$.
\end{center}
Let $x_\alpha$ be random relative to $A$ and $w_\alpha$. Let $y_\alpha$ be a real such that, for every $\gamma > 0$ and all sufficiently large $r$,
\begin{center}
$K^{A,w_\alpha,x_\alpha}_{(1 + \gamma) r, r}(y_\alpha \, | \, y_\alpha) = (s_1 + \frac{\delta}{2}) \gamma r = K_{(1 + \gamma) r, r}(y_\alpha \, | \, y_\alpha)$.
\end{center}
Let $z_\alpha$ be the real whose binary expansion is given by
\begin{align*}
z_\alpha[r] = \begin{cases}
w_\alpha[r] &\text{ if } a_n < r \leq \frac{1-s_2}{1-s_1}b_n \text{ for some }  n \in \N \\
x_\alpha[r] &\text{ if } \frac{1-s_2}{1-s_1}b_n < r \leq b_n  \text{ for some }  n \in \N \\
y_\alpha[r] &\text{ if } b_n < r \leq (1-\delta)a_{n+1} < \text{ for some }  n \in \N \\
x_\beta &\text{ if } (1-s_1\delta/2)a_{n+1} < r \leq a_{n+1} < \text{ where  }  f(\beta) = n \\
\end{cases}
\end{align*}
Finally, we define our set $E = \{z_\alpha\}$. 

We begin by collecting relevant aspects of our construction. Let $\alpha$ be an ordinal, let $A = A_\alpha$ be the corresponding oracle in the order, and let $z = z_\alpha$ be the point constructed at ordinal $\alpha$. 

Let$n$ be sufficiently large. Let $a_n < r \leq \frac{1-s_2}{1-s_1}b_n$,
\begin{align}
K^A_{r}(z) &= K^A_{a_n}(z) + K^A_{r,a_n}(w_\alpha \, | \, z)\tag*{}\\
&= K^A_{a_n}(z) + (r - a_n) s_1 + O(\log r)\label{eq:conditionalWalpha}.
\end{align}

Let $t = \frac{1-s_2}{1-s_1}b_n < r \leq b_n$,
\begin{align}
K^A_{r}(z) &= K^A_t(z) + K^A_{r,t}(x_\alpha \, | \, z)\tag*{}\\
&= K^A_t(z) + (r - t) + O(\log r)\tag*{}\\
&= (t - a_n)s_1+ (r - t) + O(\log r)\tag*{}\\
&= ts_1 + r -t +O(\log r)\tag*{}\\
&= r - (1-s_1)t+O(\log r)\tag*{}\\
&= r - (1-s_2)b_n + O(\log r)\label{eq:conditionalX}.
\end{align}

Let $b_n < r \leq (1-s_1\delta/2)a_{n+1}$. Then,
\begin{align}
K^A_{r}(z) &= K^A_{b_n}(z) + K^A_{r,b_n}(y_\alpha \, | \, z)\tag*{}\\
&= b_n - (1-s_2)b_n  + K^A_{r,b_n}(y_\alpha \, | \, z)\tag*{}\\
&= s_2b_n  + K^A_{r,b_n}(y_\alpha \, | \, z)\tag*{}\\
&= s_2b_n  + (s_1 + \frac{\delta}{2})(r - b_n)\label{eq:conditionalY}.
\end{align}

Finally, let $t = (1-s_1\delta/2)a_{n+1} < r \leq a_{n+1}$ and let $\beta < \alpha$ be the ordinal such that $f(\beta) = n$. Then,
\begin{align}
K^A_{r}(z ) &= K^A_{t}(z) + K^A_{r,t}(x_\beta \, | \, z)\tag*{}\\
&= s_2b_n  + (s_1 + \frac{\delta}{2})(t - b_n)  + K^A_{r,t}(x_\beta \, | \, z)\tag*{}\\
&= (s_1 + \frac{\delta}{2})t   + K^A_{r,t}(x_\beta \, | \, z)\label{eq:conditionalXbeta}
\end{align}
In particular,
\begin{align}
s_1 r &\leq K^A_{r}(z )\tag*{}\\
&= (s_1 + \frac{\delta}{2})t   + K^A_{r,t}(x_\beta \, | \, z)\tag*{}\\
&\leq (s_1 + \frac{\delta}{2})r   + r-t\tag*{}\\
&\leq (s_1 + \frac{\delta}{2})r   + \frac{\delta r}{2}\tag*{}\\
&= s_2 r\label{eq:boundonZCloseToBn}
\end{align}

Let $a_n < r \leq a_{n+1}$. The above inequalities show that, if $r > \frac{1-s_2}{1-s_1}b_n$, then
\begin{center}
$K^A_r(z) \geq s_1 r$.
\end{center}
When $r \leq \frac{1-s_2}{1-s_1}b_n$, by combining equality (\ref{eq:conditionalWalpha}) and inequality (\ref{eq:boundonZCloseToBn}),
\begin{align*}
K^A_{r}(z) &= K^A_{a_n}(z) + (r - a_n) s_1 + O(\log r)\\
&\geq s_1a_n + (r - a_n) s_1 + O(\log r)\\
&= s_1 r+ O(\log r).
\end{align*}
Therefore, $K^A_r(z) \geq s_1 r$. for all $r$. For the lower bound, let $r = \frac{1-s_2}{1-s_1}b_n$. Then,
\begin{align*}
K^A_{r}(z) &=K^A_{a_n}(z) + (r - a_n) s_1 + O(\log r)\\
&\leq a_n + + (r - a_n) s_1 + O(\log r)\\
&\leq s_1 r+ O(\log r),
\end{align*}
and so $\dim^A(z) = s_1$.

Similarly, the above inequalities show that  $K^A_r(z) \leq s_2 r$. To prove the lower bound, let $r = b_n$. Then
\begin{align*}
K^A_r(z) &=r - (1-s_2)b_n + O(\log r)\\
&= s_2 r + O(\log r),
\end{align*}
and so $\Dim^A(z) = s_2$.

To complete the proof, we must show that $E$ does not have an optimal Hausdorff oracle, nor an optimal packing oracle. Let $A = A_\alpha$ be any Hausdorff oracle for $E$. Let $B$ be an oracle encoding the set $\{w_\beta, x_\beta, y_\beta \, | \, \beta \leq \alpha\}$. Note that we can encode this information since this set is countable. 

Let $z_\beta \in E$. First, suppose that $\beta \leq \alpha$. Then by our choice of $B$, $\dim^{A_\alpha, B}(z_\beta) = 0$. So then suppose that $\beta > \alpha$. Let $n$ be a sufficiently large natural such that $f(\alpha) = n$. Then, since $x_\alpha$ is random relative to $A_\alpha$
\begin{align*}
K^{A_\alpha}_{a_{n+1}}(z_\beta) &= (s_1 + \frac{\delta}{2})t   + K^A_{a_{n+1},t}(x_\alpha \, | \, z)\\
&\geq (s_1 + \frac{\delta}{2})t  + a_{n+1} -t,
\end{align*}
where $t = (1-s_1\delta / 2)a_{n+1}$.
However, by our construction on $B$, 
\begin{align*}
K^{A_\alpha, B}_{a_{n+1}}(z_\beta) &= (s_1 + \frac{\delta}{2})t   + K^{A,B}_{r,t}(x_\alpha \, | \, z)\\
&\geq (s_1 + \frac{\delta}{2})t  + O(1).
\end{align*}
Therefore, 
\begin{align*}
K^{A_\alpha}_{a_{n+1}}(z_\beta) - K^{A_\alpha, B}_{a_{n+1}}(z_\beta) &= a_{n+1} -t \\
&= \frac{s_1\delta a_{n+1}}{2}.
\end{align*}
Since $z_\beta$ was arbitrary, it follows that $B$ reduces the complexity of every point $z \in E$ infinitely often. Since $A_\alpha$ was arbitrary, we conclude that $E$ does not have optimal Hausdorff nor optimal packing oracles.
\end{proof}

\begin{cor}
Assuming CH and AC, for every $0 < s_1 < s_2 \leq 1$ there is a set $E \subseteq \R$ which has optimal Hausdorff oracles but does not have optimal packing oracles such that
\begin{center}
$\dim_H(E) = s_1$ and $\dim_P(E) = s_2$.
\end{center}
\end{cor}
\begin{proof}
Let $F$ be a set such that $\dim_H(F) = \dim_P(F) = s_1$. Then, by Lemma \ref{lem:regularHasOptimalHaus}, $F$ has optimal Hausdorff oracles. Let $G$ be a set, guaranteed by Theorem \ref{thm:noHausnoPackingOptimal}, with $\dim_H(G) < s_1$, $\dim_P(G) = s_2$ such that $G$ does not have optimal Hausdorff nor optimal packing oracles.

Let $E = F \cup G$. Then $\dim_H(E) = s_1$ and $\dim_P(E) = s_2$ by the union formula for Hausdorff and packing dimension. By Observation \ref{obs:optimalSubset}, $E$ has optimal Hausdorff oracles. 

We now prove that $E$ does not have optimal packing oracles. Let $A$ be a packing oracle for $E$. By possibly joining $A$ with a packing oracle for $G$, we may assume that $A$ is a packing oracle for $G$ as well. Since $G$ does not have optimal packing oracles, there is an oracle $B\subseteq \N$ and $\epsilon > s_2 - s_1$ such that, for every $x\in G$ where $\Dim^{A,B}(x) \geq s_2 - \epsilon$,
\begin{center}
$K^{A,B}_r(x) < K^A_r(x) - \epsilon r$
\end{center}
for infinitely many $r$. Let $x \in E$ such that $\Dim^{A,B}(x) \geq s_2 - \epsilon$. Then, by our choice of $F$, $x$ must be in $G$. Therefore 
\begin{center}
$K^{A,B}_r(x) < K^A_r(x) - \epsilon r$
\end{center}
for infinitely many $r$, and so $A$ is not an optimal packing oracle for $E$. Since $A$ was arbitrary, the conclusion follows.

\end{proof}

\begin{thm}
Assuming CH and AC, for every $0 < s_1 < s_2 \leq 1$ there is a set $E \subseteq \R$ which has optimal packing oracles but does not have optimal Hausdorff oracles such that
\begin{center}
$\dim_H(E) = s_1$ and $\dim_P(E) = s_2$.
\end{center}
\end{thm}
\begin{proof}
Let 
\begin{center}
$F = \{ x \in (0,1) \, | \, \dim(x) = 0 \text{ and } \Dim(x) = s_2\}$.
\end{center}
By Lemma \ref{lem:OscillatingHasOptimal}, $\dim_H(F) = 0$, $\dim_P(F) = s_2$ and $F$ has optimal packing oracles. Let $G$ be a set, guaranteed by Theorem \ref{thm:noHausnoPackingOptimal}, with $\dim_H(G) = s_1$, $\dim_P(G) = s_2$ such that $G$ does not have optimal Hausdorff nor optimal packing oracles. Let $E = F \cup G$. Then $\dim_H(E) = s_1$ and $\dim_P(E) = s_2$ by the union formula for Hausdorff and packing dimension. By Observation \ref{obs:optimalPackingSubset}, $E$ has optimal packing oracles. 

We now prove that $E$ does not have optimal Hausdorff oracles. Let $A$ be a Hausdorff oracle for $E$. By possibly joining $A$ with a Hausdorff oracle for $G$, we may assume that $A$ is a Hausdorff oracle for $G$ as well. Since $G$ does not have optimal Hausdorff oracles, there is an oracle $B\subseteq \N$ and $\epsilon > s_1$ such that, for every $x\in G$ where $\dim^{A,B}(x) \geq s_1 - \epsilon$,
\begin{center}
$K^{A,B}_r(x) < K^A_r(x) - \epsilon r$
\end{center}
for infinitely many $r$. Let $x \in E$ such that $\dim^{A,B}(x) \geq s_1 - \epsilon$. Then, since $\dim_H(F) = 0$, $x$ must be in $G$. Therefore 
\begin{center}
$K^{A,B}_r(x) < K^A_r(x) - \epsilon r$
\end{center}
for infinitely many $r$, and so $A$ is not an optimal Hausdorff oracle for $E$. Since $A$ was arbitrary, the conclusion follows.
\end{proof}

\section{Acknowledgments}
I would like to thank  Denis Hirschfeldt, Jack Lutz and Chris Porter for very valuable discussions and suggestions. I would also like to thank the participants of the recent AIM workshop on Algorithmic Randomness.

\bibliography{optimal}
\end{document}